\documentclass[10pt,a4paper]{article}

\usepackage{epsfig}
\usepackage{microtype}

\usepackage{amsmath}

\usepackage{amssymb}
\usepackage{amsthm}
\usepackage{thmtools}
\usepackage{thm-restate}
\usepackage{relsize}
\usepackage{mathtools}
\usepackage{titlesec}
\usepackage{fancybox}
\usepackage{xcolor}
\usepackage{xspace}
\usepackage{enumerate}
\usepackage{graphicx}
\usepackage{epstopdf}

\usepackage{cleveref}
\usepackage{algorithm}
\usepackage{paralist}
\usepackage{url}

\usepackage{a4wide}

\usepackage{authblk}

\renewcommand{\epsilon}{\varepsilon}

\newcommand{\myurl}[1]{{\footnotesize \url{#1}}}

\newcommand{\donotshow}[1]{}

\newcommand{\ignore}[1]{}

\newcommand{\proofendswithequation}{\vspace*{-\baselineskip}\vspace*{-\belowdisplayskip}\par}

\newcommand{\argmin}{\mathop {\rm argmin}}
\providecommand{\R}{\mathbb{R}}
\newcommand{\Rplus}{\R_{\ge 0}}

\newcommand {\abs}[1] {| #1 |}

\newcommand{\set}[2]{ \left\{\, #1 \mbox{ {\rm :} } #2 \,\right\}  }
\providecommand{\set}[2]{ \left\{\, #1 \,\, \mbox{{\rm :}}\,\, #2 \, \right\}  }

\newcommand{\norm}[1]{\vert\!\vert #1 \vert\!\vert}
\newcommand{\onenorm}[1]{{\norm{#1}_{1}}}

\newcommand{\diag}{\mathrm{diag}}

\newcommand{\E}{\mathcal{E}}

\DeclareMathOperator{\supp}{{\mathrm{supp}}}

\makeatletter
\renewcommand{\paragraph}{%
 \@startsection{paragraph}{4}%
 {\z@}{1.2ex \@plus .5ex \@minus .1ex}{-.7em}%
 {\normalfont\normalsize\bfseries}%
}
\makeatother

\newtheorem{lemma}{Lemma}

\newtheorem{theorem}{Theorem}
\newtheorem{fact}{Fact}
\newtheorem{assumption}{Assumption}

\newcommand{\lref}[1]{(\ref{#1})}
\renewcommand{\epsilon}{\varepsilon}

\newcommand{\normtwo}[1]{\norm{#1}_2}
\newcommand{\twonorm}[1]{\normtwo{#1}}

  \newcommand{\V}{{\cal V}}

\newcommand{\F}{{\cal{F}}}
\newcommand{\C}{{{\cal{C}}}}
\renewcommand{\L}{{{\cal{L}}}}

\newcommand{\Tr}{{\mathrm{Tr}}}

\renewcommand{\Im}{{\mathrm{Im}}}
\newcommand{\Ker}{{\mathrm{Ker}}}

\newcommand{\row}[2][]{#2_{#1}}
\newcommand{\col}[2][]{#2^{#1}}
\newcommand{\entry}[2][]{#2_{#1}}

  \title{Physarum-Inspired Multi-Commodity Flow Dynamics}
  \author[1]{Vincenzo Bonifaci}
  \author[2]{Enrico Facca}
  \author[3]{Frederic Folz}
  \author[4]{\\Andreas Karrenbauer}
  \author[4]{Pavel Kolev}
  \author[4]{Kurt Mehlhorn} \author[3]{\\Giovanna Morigi}
 \author[5]{Golnoosh Shahkarami}
  \author[6]{Quentin Vermande}
 
  \affil[1]{Dipartimento di Matematica e Fisica, Universit{\`a} degli Studi Roma Tre, Roma, Italy} 
  \affil[2]{Scuola Normale Superiore, Pisa, Italy} \affil[3]{Fachbereich Physik, Universit\"at des Saarlandes, Saarbr\"ucken, Germany}
  \affil[4]{Max Planck Institute for Informatics, Saarbr\"ucken, Germany}
  \affil[5]{Max Planck Institute for Informatics and Fachbereich Informatik, Universit\"at des Saarlandes}
 \affil[6]{\'Ecole Normale Sup\'erieure, Paris, France}
\begin{document}


\maketitle

\begin{abstract} In wet-lab experiments, the slime mold Physarum polycephalum has demonstrated its ability to tackle a variety of computing tasks, among them the computation of shortest paths and the design of efficient networks. For the shortest path problem, a mathematical model for the evolution of the slime is available and it has been shown in computer experiments and through mathematical analysis that the dynamics solves the shortest path problem. In this paper, we generalize the dynamics to the network design problem. We formulate network design as the problem of constructing a network that efficiently supports a multi-commodity flow problem. We investigate the dynamics in computer simulations and analytically. The simulations show that the dynamics is able to construct efficient and elegant networks. In the theoretical part we show that the dynamics minimizes an objective combining the cost of the network and the cost of routing the demands through the network. We also give alternative characterizations of the optimum solution. 
\end{abstract}




\section{Introduction}

Physarum polycephalum is a slime mold in the Mycetozoa group~\cite{Baldauf:1997}. Its cells can grow to considerable size and it can form networks.
In wet-lab experiments, the slime mold Physarum polycephalum was applied to a diverse variety of computing problems: computation of shortest paths in a network~\cite{Nakagaki-Yamada-Toth}, computation of minimum risk paths~\cite{PhysarumMinimumRiskPath},
design of efficient networks~\cite{Tero-Takagi-etal,Adamatzky-Belgium,Watanabe-Tero2011,Adamatzky-Martinez}, computation of Voronoi and Delaunay diagrams~\cite{AdamatzkyProximity,Adamatzky-Voronoi-Delauney}, computing circuits and electronics~\cite{Jones-Adamatzky} and many more. We refer the reader to~\cite{Adamatzky38,PhysarumBook} for a survey of the many problems that can be attacked using live Physarum polycephalum and for which the slime is able to find good or even optimal solutions to instances of limited size.  Figure~\ref{Wet-Lab Experiments} illustrates the shortest path and the network formation experiments in~\cite{Nakagaki-Yamada-Toth,Tero-Takagi-etal}.

There is also considerable work aimed at understanding the inner workings of Physarum polycephalum, for example, how global synchronisation can result from random peristaltics~\cite{Alim2013}, how information can be transported and a memory can exist in an organism without a nervous system~\cite{Alim2017,Alim2021}, and whether tubes of the mold can transfer electricity~\cite{Whiting}.

It is important to stress that the plasmodium of Physarum polycephalum is not an automaton.\footnote{There is a small community of researchers that think differently, see~\cite{Haughness} for example.}  The papers~\cite{Mayne:2018,Miyaji:2008} clearly demonstrate the limits of slime-mold computations even for the shortest path problem and~\cite{PhysarumBook,Mayne:2018} argue convincingly that conventional computing terminology should be applied with great care when discussing biological systems. After all, the solutions  constructed in wet-lab experiments strongly depend on the initial conditions, e.g., how much food is provided and how the plasmodium is distributed initially, the solutions are not strictly optimal but only approximately optimal, the outcomes of the experiments are not deterministic and hence hard to reproduce, and the maze in Figure~\ref{Wet-Lab Experiments} has fairly narrow edges and hence guides the slime towards building nearly straight connections. The papers~\cite{Stepney,Horstman-Stepney} discuss more generally the question what it means for a biological or physical device to compute. According to their definition, to which we subscribe, Physarum polycephalum does not compute. A comprehensive survey of analog computing models is given in~\cite{Bournez2018} and~\cite{Navlakha2011,Navlakha2015} discuss the differences and commonalities of biology and computing. 

\begin{figure}[ht]
  \begin{center} \includegraphics[height=6cm]{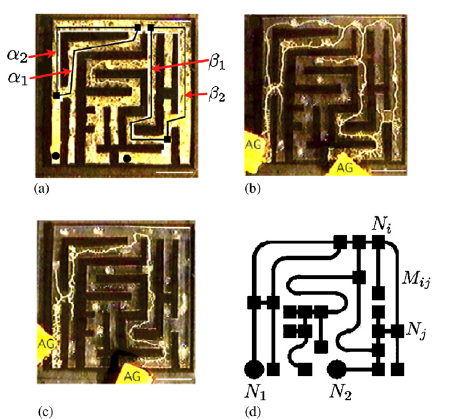}\quad\includegraphics[height=6cm]{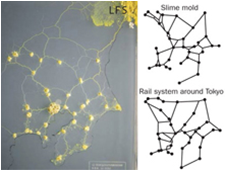}
  \end{center}
 \caption{\label{Wet-Lab Experiments} The figure on the left shows the shortest path experiment. It is reprinted from~\cite{Nakagaki-Yamada-Toth}. The edges of a graph were uniformly covered with Physarum (subfigure (a)) and food in the form of oat-meal was provided at the locations labeled AG in subfigure (b). After a while the slime retracted to the shortest path connecting the two food source (subfigure (c)). The underlying graph is shown in (d).\newline The figure on the right shows the network design experiment. It is reprinted from~\cite{Tero-Takagi-etal}. Food was provided at many places (the larger dots in the picture) and the slime was constrained to live in an area that looks similar to the greater Tokyo region. The large dot in the center corresponds to Tokyo and the empty region below it corresponds to Tokyo bay. The slime formed a network connecting the food sources. The two graphs on the right compare the network built by the slime with the railroad network in the Tokyo region. }
\end{figure}

The experimental work mentioned above instigated the development of \emph{Physarum-inspired algorithms} mimicking (parts of) the behavior of the slime mold. This is akin to algorithms mimicking ant colonies~\cite{Dorigo}, neural networks~\cite{Aggarwal-NeuralNetworks}, simulated annealing~\cite{Kirkpatrick}, and other bio-inspired computing paradigms. Physarum-inspired algorithms have been used to solve a variety of computational tasks, for example, the design of transportation networks~\cite{Tero-Takagi-etal,Watanabe-Tero2011,Yang-Mayne-Deng} and supply-chain networks~\cite{Zhang-Adamatzky2017}. For the shortest path problem, a mathematical model in the form of a coupled system of differential equations was given for the evolution of the slime, the biological relevance of the model was argued, and the model was shown to solve shortest path problems in computer simulations~\cite{Tero-Kobayashi-Nakagaki}. Mathematical proofs that the model solves (it is explained below what is meant by solves) shortest path problems can be found in~\cite{Miyaji-Ohnishi,Physarum}. The Physarum dynamics is also able to solve more general linear programs~\cite{SV-IRLS,SV-LP,Ito-Convergence-Physarum,Karrenbauer-Kolev-Mehlhorn:NonUniformPhysarum}. It is important to emphasize that the Physarum dynamics is inspired by the behavior of the mold, it captures -- at best -- parts of the behavior of the mold. 

\emph{The paper~\cite{Tero-Kobayashi-Nakagaki} is the starting point for this work.} Tero et al.\ model the slime network as an electrical network $G=(V,E)$  with time varying resistors. Each edge $e$ of the network has a fixed positive \emph{length} $\entry[e]{c}$ and a time-varying diameter $x_e(t)$. In this paper, we will refer to $c_e$ as the \emph{cost} of the edge and to $x_e$ as the \emph{capacity} of the edge. The \emph{resistance} of $e$ at time $t$ is then $r_e(t) = c_e/x_e(t)$. Let $s_0$ and $s_1$ be two fixed vertices in the network; they represent the two food sources.  One unit of electrical current is sent from $s_0$ to $s_1$. Let $q_e(t)$ be the current flowing across $e$. Then the capacity of $e$ evolves according to the differential equation
\begin{equation}\label{dynamics}  \dot{x}_e = \frac{d}{dt} x_e (t) = \abs{q_e} - x_e \quad\text{for all $e \in E$}, \end{equation}
i.e., the diameter of an edge grows (shrinks, stays unchanged) if the flow is larger than (smaller than, equal to) the current diameter. The equations for the different edges are coupled because the flow through an edge $e$ depends on the resistance of all other edges. As is customary, we write $\dot{x}$ for the derivative with respect to time and drop the time-argument of $x$ and $q$. Tero et al.~showed in computer simulations that a discretization of the model converges to the shortest path connecting the source and the sink in the following sense: $x_e(\infty) = 1$ for the edges on the shortest path and $x_e(\infty) = 0$ for the other edges. This assumes that the shortest path is unique. Bonifaci et al.~\cite{Physarum} proved that the dynamics converges to the shortest source-sink path, i.e., that the Physarum dynamics \lref{dynamics} solves the shortest path problem. A related dynamics \begin{equation*} \dot{x}_e = \abs{q_e}^\mu - x_e \quad\text{for all $e \in E$}, \end{equation*}
where $\mu$ is a constant larger than one, may converge to a path different from the shortest path depending on the initial conditions~\cite{Miyaji:2008}.

In this paper, we generalize the model of Tero et al.~\cite{Tero-Kobayashi-Nakagaki} to network design. We introduce a simple model, again in the form of a system of differential equations, 
  \begin{itemize}
  \item that for the case of the shortest path problem agrees with the model proposed in Tero et al.,
  \item that in computer simulations qualitatively reconstructs the behavior observed in the wet-lab network design experiments~\cite{Nakagaki-Yamada-Toth,Tero-Takagi-etal}, and
  \item that is amenable to theoretical analysis.
  \end{itemize}
  We do not argue biological plausibility and we do not claim any biological relevance. We also do not try to describe a general model of the Physarum that fits all experimental setups; less ambitiously, we focus on the network design experiments. This is a paper in algorithm design and analysis. 

The shortest path problem can be viewed as a network design problem. Given two vertices in a graph, the goal is to construct the cheapest network connecting the given vertices. The solution is the shortest path connecting the vertices. The shortest path problem can also be viewed as a minimum cost flow problem. We want to send one unit of flow between the given vertices and the cost of sending a certain amount across an edge is equal to the cost of the edge times the amount sent. The solution is the shortest path connecting the given vertices. 

Networks are designed for a particular purpose. For this paper, the purpose is multi-commodity flow. Suppose that we have many pairs of vertices between which we want to send flow.  We want to construct a network that satisfies the many demands in an economical way. Economical could mean many things: minimum cost of the network (that's the Steiner tree problem), shortest realization of each demand (then the network is the union of the shortest paths), or something in the middle, i.e., some combination of the total cost of the network and the cost of routing the demands in the network. We assume \emph{economies of scale}, i.e., that there is some benefit in sharing a connection, i.e., the cost of sending one unit each of two commodities across an edge is lower than two times the cost of sending one unit of one commodity across the edge. In Section~\ref{Case Studies}, we give examples of how sharing is encouraged by our model. The principles of our model are simple. As in Tero et al.\ each edge has a cost and a capacity. We have demands between pairs of vertices; this could be passengers entering the network at some station and leaving the network at some other station. The demand $i$ between vertices $s^{(1)}_i$ and $s^{(2)}_i$ leads to an electrical flow $q^i$. For each edge $e$, we aggregate the individual flows $q^i(e)$ to an overall flow $q(e)$. This flow is then used in equation~\lref{dynamics}. For the aggregation we use either the one-norm or the two-norm of the vector $(q^1(e), q^2(e), \ldots)$ and find that the two-norm aggregation is to be prefered. We mention that one-norm aggregation is used in~\cite{Watanabe-Tero2011}.

This paper is organized as follows. In Section~\ref{The Model} we introduce our model and in Section~\ref{Results} we review our results. In Section~\ref{Related Work} we discuss related work. In Section~\ref{Case Studies}, we report about paper-and-pencil and computer experiments. The analytical part starts with Section~\ref{Preliminaries}. We review basic facts about electrical flows. In subsequent sections, we prove the existence of a solution defined for $t \in [0,\infty)$, characterize the fixed points, introduce a Lyapunov function $\L$ for the dynamics, derive further properties of the Lyapunov minimum, show convergence to the Lyapunov minimum, and finally make a connection to mirror descent. Section~\ref{conclusions} offers conclusions.

\renewcommand{\pi}{p^{i}}
\newcommand{\qi}{q^{i}}
\newcommand{\bi}{b^{i}}
\newcommand{\iof}[2][]{{#2}^{#1}}
\newcommand{\mybi}[1]{\iof[#1]{b}}
\newcommand{\fluxi}{f^{i}}

\section{The Model}\label{The Model}

Before presenting our model we want to fix some notations.  Given a
matrix $M \in \R^{n\times m}$, we will denote with
$\row[i]{M}$, $\col[j]{M}$, and $\entry[i,j]{M}$ its $i^{th}$-row, its
$j^{th}$-column, and its entry $i,j$, respectively. Moreover, we
denote with $\Tr(M)=\sum_{i=1}^{n}\entry[i,i]{M}$ the trace of a square
matrix $M\in \R^{n,n}$.

Our model for the multi-commodity network design problem is inspired
by the Physarum model for the shortest path problem and its generalization to linear programming. Let $A \in
\R^{n\times m}$ be an arbitrary real
matrix  and let $\iof[1]{b}$ to $\iof[k]{b}$ in $\R^n$ be $k$ right-hand sides such that each of the linear systems $A f = \iof[i]{b}$ is solvable. 

The reader may want to think of $A$ as
  the node-arc incidence matrix of a connected undirected graph $G$
  with $n$ nodes and $m$ edges, i.e., for each $e = (u,v) \in E$, the
  column $\row[e]{(A^T)}$ has an entry $+1$ in position $u$ and entry $-1$ in
  position $v$; the orientation of the edge is arbitrary, but fixed.
  We have $k$ different source-sink pairs $(s^{(1)}_i, s^{(2)}_i)$, $1
  \le i \le k$. Let $\iof[i]{b} \in R^n$ be the vector with entry $+1$ in
  position $s^{(1)}_i$ and entry $-1$ in position $s^{(2)}_i$. All
  other entries of $\iof[i]{b}$ are zero.  Since $G$ is assumed to be
  connected, the linear system $A f = \iof[i]{b}$ admits solutions for all
  $i$. We refer to this setting as the \emph{multi-commodity flow setting}. 

  Now, for any non-negative vectors $x \in \R^m$ and $c \in \R^m$, we define the
following matrices
\begin{equation}
  X = \diag(x)
  \quad
  C = \diag(c)
  \quad
  L(x)=A X C^{-1} A^T.
\end{equation}
Given a solution $f$ of $Af = \iof[i]{b}$, we use
\[
E_x(f) = \begin{cases}   \sum_e c_e/x_e f_e^2  & \text{if $\supp f \subseteq \supp x$,}\\
  \infty & \text{if $\supp f \not\subseteq \supp x$.}
\end{cases}
\]
to denote the \emph{energy} of $f$ with respect to $x$.
Let $\iof[i]{q}(x)  \in \R^m$, or simply $\iof[i]{q}$, be the minimum energy solution
i.e.,
\newcommand{\Argmin}{\operatornamewithlimits{argmin\vphantom{q}}}
\begin{equation}\label{def of q}
  \qi=\Argmin_{f \in \R^{m}} \left\{ E_x(f) \, : \, A f = \iof[i]{b} \right\}.
\end{equation}
 The optimal solution of the optimization problem (\ref{def of q}) (see
 Section~\ref{Preliminaries} for details) is given by
\begin{equation}
  \qi(x)=X C^{-1} A^T \iof[i]{p}(x),\quad i=1,\ldots,k,
\end{equation}
where $\iof[i]{p}(x)$, or simply $\iof[i]{p}$, is defined as any solution to
\[     L(x) \iof[i]{p} = \iof[i]{b},\quad i=1,\ldots,k.\]
In the multi-commodity flow setting, the minimal energy solution $\qi$ is simply the electrical flow realizing the demand $\bi$ and $\pi$ are the corresponding node potentials. The node potentials are not unique; they can be made unique by defining a particular node as ground, i.e., giving it potential zero. The electrical flow is induced 
by the potential drops $A^T \iof[i]{p}$ multiplied by the conductivity $X C^{-1}$.  If we now define the matrix $B\in
\R^{n,k}$ by
\[
B=\left(\col[1]{B},\ldots,\col[k]{B}\right):=\left(\iof[1]{b},\ldots,\iof[k]{b}\right)\;, 
\]
we can express the potentials, the potential drops per unit cost, and the fluxes corresponding
to the different commodities in the following matrix form
\begin{equation}
\begin{aligned}
  P&=\left(\col[1]{P},\ldots,\col[k]{P}\right)&&:=\left(\iof[1]{p},\ldots,\iof[k]{p}\right)&&\text{ with } L(x) P = B\\
\Lambda&=\left(\col[1]{\Lambda},\ldots,\col[k]{\Lambda}\right)& &:=\left(\lambda^{1},\ldots,\lambda^{k}\right)&&=C^{-1} ATP\\
Q&=\left(\col[1]{Q},\ldots,\col[k]{Q}\right)&&:=\left(q^{1},\ldots,q^{k}\right)&&=X \Lambda.
\end{aligned}
\end{equation}
Note that we use $\col[i]{P}$ and $\iof[i]{p}$ interchangeably and similarly for $\Lambda$ and $Q$. 

 We are now ready to define our model. We let the
  vector $\row[e]{Q}$ of values $\entry[e,i]{Q}$ for any edge $e$ determine the
  capacity of an edge and study different ways of combining the
  individual solutions, in particular, one-norm and two-norm\footnote{In the multi-commodity flow setting, the $q_i$'s are flows in the network $G$. The fact that flows from different demand pairs on the same edge do not cancel each other (not even partially) seems a bit strange at the microscopic level. After all, physically, only the cytoplasm is being transported. How does an edge "distinguish" between the cytoplasm of pair $i$ and the cytoplasm of pair $i'$? For this reason, we do not claim biological plausibility for our model. When $k=1$, clearly this was not an issue.}.
%
This leads to the following dynamics:
\begin{align}
  \dot{x}_e &=  -x_{e}+\sum_{i}\abs{\entry[e,i]{Q}}=x_{e}\left(-1+\sum_{i}\frac{\abs{\entry[e,i]{Q}}}{x_{e}}\right)=x_{e}\left(-1+\sum_{i}\abs{\entry[e,i]{\Lambda}}\right)=x_{e}(\onenorm{\Lambda_{e}}-1), \label{one-norm dynamics}\\
  \dot{x}_e &= -x_{e}+\sqrt{\sum_{i}\entry[e,i]{Q}^{2}}=x_{e}\left(\sqrt{{\sum}_{i}\left(\frac{{\entry[e,i]{Q}}}{x_{e}}\right)^{2}}-1\right)=x_{e}\left(\sqrt{{\sum}_{i}\entry[e,i]{\Lambda}^{2}}-1\right)=x_{e}(\normtwo{\Lambda_{e}}-1).
  \label{two-norm dynamics}
\end{align}
In~\lref{one-norm dynamics}, we form the one-norm $\onenorm{\Lambda_e}$ of the different normalized potential drops across  any edge $e$, and in~\lref{two-norm dynamics}, we form the two-norm $\twonorm{\row[e]{\Lambda}}$. 
For $k = 1$, the one-norm and the two-norm dynamics coincide. The results of this paper suggest that the two-norm dynamics is the appropriate generalization to larger $k$.

The following \emph{generalized Physarum dynamics} introduced in~\cite{Bonifaci16-Generalized-Physarum} subsumes the two-norm
dynamics as a special case. For each $e \in E$, let $g_e$ be a non-negative, increasing and differentiable function with $g_e(1) = 1$: 
\begin{equation}\label{eq:GenDyn}
	\dot{x}_{e} = x_e \left(g_e\left(\twonorm{\Lambda_e}\right)  - 1\right). 
      \end{equation}
      The two-norm dynamics
      is a special case with $g_e(z) = 1 + \left( z- 1\right)$.
      Other examples are $g_e(z) = 1 + d_e\left( z - 1\right)$ and $g_e(z) = 1 + d_e\left( z^2 - 1\right)$ where $d_e > 0$ is the ``reactivity''~\cite{Karrenbauer-Kolev-Mehlhorn:NonUniformPhysarum} of edge $e$, $g_e(z) = z^{\mu_e}$ for some $\mu_e > 0$ or $g_e(z) = (1 + \alpha_e) z^{\mu_e}/(1 + \alpha_e z^{\mu_e})$ for some $\mu_e, \alpha_e > 0$.

      The right hand sides of~\lref{one-norm dynamics} to~\lref{eq:GenDyn} are defined for any
      \[ x \in \Omega = \set{
        x \in \R^m_{\ge 0}}{\exists P\in\R^{n,k} {\text{ solving } } 
        L(x) P=B}.\]

      \section{Our Results}\label{Results}

      In the analytical part of the paper, we ask and answer the
      following questions for the generalized Physarum dynamics. We
      have little to say about the one-norm dynamics.
      \begin{itemize}
        \item Does the dynamics have a solution $x(t)$ with $t \in [0,\infty)$?
      \item Does the dynamics converge?
      \item What are the fixed points and the limit points of the dynamics?
        \item What does the dynamics optimize?
      \item How can we characterize the limit points?
      \end{itemize}
      In the experimental part of the paper, we perform computer and pencil-and-paper simulations of the dynamics and address the following questions:
      \begin{itemize}
      \item How strong are the sharing effects of the dynamics? How far deviate individual flows from their shortest realization in order to benefit from sharing edges with other flows?
        \item Does the dynamics construct ``nice'' networks? Does it qualitatively reconstruct the wet-lab experiments in~\cite{Tero-Takagi-etal}?
        \end{itemize}

Our first result concerns the existence of solutions with domain $[0,\infty)$ for the generalized Physarum dynamics.

        \begin{theorem}\label{thm: Existence}Let $x(0) \in \R^m_{> 0}$. The generalized Physarum dynamics has a solution $t \mapsto x(t) \in \R^m_{> 0}$ for $t \in [0,\infty)$. \end{theorem}

The \emph{cost of a capacity vector} $x$ is defined as
        \[    \C(x) = c^T x = \sum_e c_e x_e.\]
        The \emph{energy dissipation} for a single demand $b$ induced by a capacity vector $x$ is defined as
        \[  \min_{f;\ Af = b} E_x(f) = \sum_e r_e/x_e q_e^2 = b^Tp = p^T L(x) p,\]
        where $q$ is the minimum energy solution of $Af = b$ with respect to $x$ and $p$ is the corresponding node potential. We will show the second equality in Section~\ref{Preliminaries}. The last equality follows from $L(x) p = b$. The \emph{energy dissipation} $\E(x)$ for a set of $k$ demands $\iof[1]{b}$, \ldots, $\iof[k]{b}$ is the sum of the energy dissipations for the individual demands, i.e.,
        \[  \E(x) = \sum_i E_x(\qi)= \sum_i (\iof[i]{b})^T \iof[i]{p} = \Tr(P^T L(x) P), \]
where $\iof[i]{p}$ is the node potential with respect to the minimum energy solution $\qi$ to the $i$-th demand. 

The \emph{fixed points of a dynamics} are the points $x$ with $\dot{x}_e = 0$ for all $e$. We use $\F_1$ and $\F_g$ to denote the fixed points (also called equilibrium points) of the one-norm and the generalized dynamics. 

   \begin{lemma}[The fixed points of the one-norm dynamics] $x \in \F_1$ iff for all $e$ either $x_e = 0$ or $\onenorm{\row[e]{\Lambda}} = 1$. The latter condition is equivalent to $\onenorm{Q_e} = x_e$ as well as to $\onenorm{\row[e]{(A^T)} P} =c_{e}$. 
   \end{lemma}

   The fixed points of the generalized dynamics have a remarkable property. \emph{For a fixed point $x \in \F_g$, the cost $\mathcal{C}(x)$ equals the dissipated energy $\mathcal{E}(x)$.}

	\begin{lemma}[The fixed points of the generalized Physarum dynamics]\label{lem:FixPoints}
		$x \in \F_g$ iff for all $e$ either $x_e = 0$ or $\twonorm{\row[e]{\Lambda}} = 1$. 
		The latter condition can be expressed equivalently by $\twonorm{\row[e]{(A^T)} P} =c_{e}$ 
		and also by $x_{e}= \twonorm{\row[e]{Q}}$.
		Further, for every $x\in\mathcal{F}_{g}$ we have 
		$x\geq0$, $AQ=B$, and
		\[
			\Tr(B^{T}P)=\mathcal{E}(x) = \mathcal{C}(x)=c^{T}x, 
                      \]
                      i.e., for fixed points of the generalized Physarum dynamics the cost equals the energy dissipation. 
                    \end{lemma}

The beauty goes further. \emph{The dynamics follows a path along which the sum of cost and energy dissipation decreases and, under mild additional assumptions, minimizes the sum in the limit of $t \rightarrow \infty$}. 
Let 
\[     \L(x) = \frac{1}{2}( \C(x) + \E(x)) = \frac{1}{2} ( c^T x + \sum_i (\iof[i]{b})^T \iof[i]{p}) \]
be one-half of the sum of the cost and the enery dissipation of the network. 
We show in Section~\ref{Lyapunov Function} that the function $\L$
is a \emph{Lyapunov function} for the generalized Physarum dynamics, in particular, $\L(x(t))$ is a non-negative decreasing function of $t$. Formally, the conditions for a Lyapunov function are: $\L(x) \ge 0$ for $x \in \Omega$, $\frac{d}{dt}\L(x(t)) = \langle \nabla L, \dot{x} \rangle \le 0$ for all $t$. In the case $k = 1$, $\L$ is also a Lyapunov function for the one-norm dynamics as shown in~\cite{Karrenbauer-Kolev-Mehlhorn:NonUniformPhysarum}.
Let
\[    \V = \set{x}{\langle \nabla \L, \dot{x} \rangle = 0}\]
be the set of points in which the dynamics does not decrease the Lyapunov function any further. 
It follows from general theorems about dynamical systems that the dynamics converges to the set $\V$. We show that $\V$ is equal to the set of fixed points $\F_g$, and that under  mild additional assumptions, the dynamics converges to the minimizer of the Lyapunov function.

\begin{theorem} $\F_g = \V$ and the generalized Physarum dynamics converges to $\V$. Moreover, if the set $\F_g$ is finite and any two points in $\F_g$ have distinct values of $\L$, the dynamics $x(t)$ converges to $x^* = \argmin_{x \in \R^m_{\ge 0}} \L(x)$. \end{theorem}

The minimum of the Lyapunov function can also be characterized in alternative ways.

                         \newcommand{\MQ}{\mathit{MinQ}}
                         \newcommand{\MP}{\mathit{MaxP}}
                         \newcommand{\ML}{\mathit{MinL}}

\begin{theorem} The following quantities $\MQ$, $\MP$, and $\ML$ are equal. 
\begin{align}
  \MQ &= \min_{Q \in \R^{m \times k}} \set{\sum_e c_e  \normtwo{\row[e]{Q}} } { A Q = B },\\
  \MP &= \max_{P \in \R^{n \times k}}\set{ \Tr(B^TP)}{ \normtwo {\row[e]{(A^T)} P} \le c_e \text{ for all $e$}},\\
  \ML &= \min_{x \in \R^m_{\ge 0}}  \L(x). 
  \end{align}
  Moreover, there are optimizers $Q^*$, $R^*$ and $x^*$ such that
  \begin{align*}  x^*_e &= \normtwo{\row[e]{Q^*}} \quad\text{ for all $e$},\\
    L(x^*) P^* &= B ,\\
    Q^* &= X^* C^{-1} A^T P^*.
  \end{align*}
\end{theorem}

It is instructive to interpret the theorem for the case $k = 1$, $A$ the node-arc incidence matrix of a directed graph, and $b$ a vector with one entry $+1$ and one entry $-1$ and all other entries equal to zero. Then $\MQ = \min_{q \in \R^m} \set{\sum_e c_e \abs{q_e}}{A q = b}$ is the minimum cost of a flow realizing $b$ in the underlying undirected network and $\MP = \max_{p \in \R^n}\set{b^T p }{ \abs{p_v - p_u} \le c_e \text{ for all $e = (u,v)$}}$ is the maximum distance between the two nodes designated by $b$ for any distance function on the nodes satisfying the cost constraints imposed by $c$. Both values are equal to the cost of the minimum cost path connecting the two designated nodes and hence $\MQ = \MP$. The third characterization via 
$\ML = \min_{x \ge 0} \L(x)$ is non-standard. Note that $\L(x) = (c^T x + b^T p)/2$, where $p$ are node potentials driving a current of 1 between the nodes designated by $b$ in the network with edge resistances $c_e/x_e$. Then $b^Tp$ is the potential difference between the two designated nodes which, since the driven current is one, is the effective resistance between the two designated nodes. In Lemma~\ref{Gradient of phi}, we will show $\frac{\partial}{\partial x_{e}}\L(x) = \frac{c_{e}}{2}( 1- \twonorm{\Lambda_e}^2)$, i.e., the minimizer $x^*$ of $\L(x)$ must satisfy $x^*_e \not= 0 \Rightarrow \abs{\row[e]{(A^T)} p^*} = c_e$, where $p^*$ are node potentials corresponding to $x^*$. Note that for an edge $e = (u,v)$, $\abs{\row[e]{(A^T)} p^*} = \abs{p^*_v - p^*_u}$ is the potential drop on $e$. Orient all edges such that potential drops are positive and consider any path $W$ ($W$ for Weg) in $\supp(x^*)$  connecting the two designated nodes. Then
\[    b^T p^* = \sum_{e \in W} \row[e]{(A^T)} p^* = \sum_{e \in W} c_e, \]
since the potential difference between the two designated nodes is the sum of the potential drops along $W$. Thus any two paths in $\supp(x^*)$ connecting the two designated nodes must have the same cost and hence (assuming that any two such paths have distinct cost) $\supp(x^*)$ contains a single path connecting the two designated nodes. In fact, $\supp(x^*)$ is equal to such a path. Now $\sum_{e \in W} c_e x_e + \sum_{e \in W} c_e/x_e = \sum_{e \in W} (c_e x_e + c_e/x_e)$ is minimized for $x_e = 1$ for all $e \in W$ and then is equal to twice the cost of $W$. Of course, the cost of $W$ is minimized for the shortest undirected path connecting the two designated nodes.

We turn to the result of our computer experiments.  We performed three \emph{case studies}, two small and the third inspired by the wet-lab experiment by~\cite{Tero-Takagi-etal}. The first example (Section~\ref{Flow in a Ring}) can be treated analytically, we consider a ring with three nodes with a demand of one between any pair of nodes. We will see that a solution using all three edges is superior to a solution using only two edges. Also, we see confirmed that for fixed points of the two-norm dynamics the cost of the network and the total energy dissipation is the same. The second example (Section~\ref{Bow-Tie Graph}) concerns flow in the Bow-Tie graph shown in Figure~\ref{BowTie Simulation}. We will investigate the incentive for sharing links. In this example, the demands can share a link at the cost of increasing the distance between the terminals. We will see that sharing pays off. The third example~\ref{Tokyo Railroad} is based on the example in~\cite{Tero-Takagi-etal}. We will see that the dynamics forms nice networks similar to the networks in~\cite{Tero-Takagi-etal}.

\section{Related Work}\label{Related Work}

This paper is inspired by~\cite{Tero-Kobayashi-Nakagaki}, ~\cite{Watanabe-Tero2011}, and~\cite{Tero-Takagi-etal}. We already explained the connection to these papers in detail in the previous sections. 

Shortly after this work was posted on arXiv, a closely related paper~\cite{Lonardi2020optimal} was posted. It considers the multi-commodity transportation problem in graphs. Let $A$ be the node-arc incidence matrix of a directed graph and vectors $b_1$ to $b_k$ with $1^T b_i = 0$ for all $i$ be $k$ supply-demand vectors. Each arc of the graph has a fixed cost $c_e$ and a capacity $x_e$. This is what we called the multi-commodity flow setting in Section~\ref{The Model}. They model the interaction between the different commodities in exactly the same way as we do, i.e., for each $i$, a minimum energy solution $q_i$ is a minimum energy solution with respect to the resistances $c_e/x_e$ of the system $A q = b_i$.  The different flows on each edge $e$ are combined by forming their two-norm. The difference lies in the dynamics. The paper considers the dynamics
\begin{equation}\label{beta dynamics}  \dot{x}_e = x_e^\beta \twonorm{\Lambda_e}^2  - x_e, \end{equation}
where $\beta \in (0,2)$ is a parameter. For $\beta = 1$, this dynamics is a special case of our generalized dynamics obtained by setting $g_e(z) = 1 + (z^2 - 1)$.

The paper investigates the dynamics analytically and experimentally. For the experimental evaluation, the paper uses the Paris metro. In the analytical part, the paper shows that the fixed points satisfy
$x_e^{3 - \beta} = \twonorm{Q_e}^2$ and that the solution to the optimization problem
\[   \text{minimize} \sum_e \frac{c_e}{x_e} \twonorm{Q_e}^2 \text{ subject to } \sum_e c_e x_e^{2 - \beta} = K \text{ and } A Q = B, \]
where $K$ is a positive constant, satisfies $x_e^{3 - \beta} = C \cdot \twonorm{Q_e}^2$, where $C$ is a constant, i.e., fixed points and optimal solutions to the optimization problem exhibit the same relation between $x_e$ and $Q_e$. The paper also contains an extensive discussion of the simulation of the dynamics and of the numerical solution of the optimization problem above.

Convergence of the dynamics is not shown. However, a slight modification of the Lyapunov function used in this paper also works for their dynamics. Assume $\beta \in (0,2)$ and define 
\[
\L(x )=\frac{1}{2}\left(\frac{1}{2-\beta}c^{T}x^{2-\beta} +\sum_{i=1}^{k} (\bi)^{T} \pi \right).
\]


\begin{lemma}[Gradient of $\L$]\label{Gradient of L} For all $e \in E$,
  \begin{equation}
    \frac{\partial}{\partial x_{e}}\L(x) = \frac{c_{e}}{2}(x_e^{1 - \beta}- \twonorm{\Lambda_e}^2).\label{eq:gradL}
  \end{equation} \end{lemma}
\begin{proof} The derivative $  \frac{\partial}{\partial x_{e}} \sum_{i=1}^{k} (\bi)^{T} \pi$ is computed in Lemma~\ref{Gradient of phi} and $ \frac{\partial}{\partial x_{e}} c^T x^{2 - \beta} = (2 - \beta) c_e x_e^{1 - \beta}$. 
\end{proof}

\begin{theorem}\label{Lyapunov for beta dynamics} 
	The function $\L: \Omega \mapsto\mathbb{R}$
	is a Lyapunov function for the dynamics (\ref{beta dynamics}), i.e., $\frac{d}{dt} \L(x(t)) \le 0$ for all $t$. Let
        \[      \V = \set{x \in \Omega}{\langle \nabla \L(x), \dot{x} \rangle = 0}. \]
        Then $\V$ is equal to the fixed points of (\ref{beta dynamics}).
\end{theorem}
\begin{proof} Since $\frac{d}{dt} \L(x(t)) = \langle \nabla \L(x), \dot{x} \rangle$, we obtain
  \begin{align*} \frac{d}{dt} \L(x(t))  = \sum_e \frac{c_e}{2} (x_e^{1 - \beta} - \twonorm{\Lambda_e}^2)\cdot (x_e^\beta \twonorm{\Lambda_e}^2 - x_e) 
    = -\sum_e \frac{c_e}{2}x_e^\beta(x_e^{1 - \beta} - \twonorm{\Lambda_e})^2 \le 0.
    \end{align*}
We have equality if and only if for all $e$ either $x_e = 0$ or $\twonorm{\Lambda_e}^2 = x_e^{1 - \beta}$. Thus $x \in \V$ if and only if $x$ is a fixed point of (\ref{beta dynamics}).
\end{proof}

\begin{lemma}For fixed points $x$ of (\ref{beta dynamics}), $\sum_e c_e x_e^{2 - \beta} = \sum_i (\bi)^T p_i$. \end{lemma}
\begin{proof} \[ \sum_i (\bi)^T \pi = \sum_e \sum_i \frac{c_e}{x_e} q_{ei}^2 = \sum_e \sum_i \frac{c_e}{x_e}  (x_e  \Lambda_{ei})^2
                                 = \sum_e c_e x_e \twonorm{\Lambda_e}^2 = \sum_e c_e x_e^{2 - \beta}.\]
                               \end{proof}

We mentioned in the result section that our generalized Physarum dynamics converges to a solution for which the cost $\sum_e c_e x_e$ is equal to the dissipated energy $\sum_i (\bi)^T p_i$. The dynamics \lref{beta dynamics} allows a wider choice of equilibrium points.

\section{Case Studies}\label{Case Studies}
\subsection{Multi-commodity Flow in a Ring}\label{Flow in a Ring}

Consider a graph consisting of three vertices $a$, $b$, and $c$ and three edges connecting them into a 3-cycle. All edges have cost one and we have a demand of one between any pair of nodes. 
An equilibrium uses either two edges or three edges. 

\subsubsection{Two Edge Solution}

We will see below that, for each of the dynamics, the solution is symmetric,  i.e., both edges have the same capacity in equilibrium, say $z$. The flow across both edges is two. For each demand, the potential drop on each edge is $1/z$. So the total energy spent is $\E = 2/z + 2 \cdot 1/z = 4/z$ (one demand uses two edges for a energy dissipation of $2/z$ and two demands use one edge for a energy dissipation of $1/z$ each) and the total cost $\C = 2z$. Thus $\C + \E= 4/z + 2z$. 

\paragraph{One-Norm Dynamics:} The current across each edge is 2 and hence $z = 2$ for each of the existing edges. Thus $\C = c^T z = 4$, $\E = \sum_i (\iof[i]{b})^T \iof[i]{p} = 4/z = 2$, and $\C + \E = 6$.

\paragraph{Two-Norm Dynamics} 
The current across each edge is $1 + 1$ and hence $z = \sqrt{2}$. Thus $\C = c^T z = 2 \sqrt{2}$, $\E = \sum_i (\iof[i]{b})^T \iof[i]{p} = 4/\sqrt{2} = 2\sqrt{2}$ and $\C + \E = 4 \sqrt{2}$.  Note that $\C = 2 \sqrt{2} = \E$. This is not a coincidence as we show in Lemma~\ref{lem:GenPhyDin}.

\paragraph{Optimum:} We have $\C + \E = 4/z + 2z$. The optimum is attained for $z = \sqrt{2}$. Note that this corresponds to the equilibrium of the two-norm. This is not a coincidence as we show in Theorem~\ref{lem:GenEqSet}.

\subsubsection{Three Edge Solution}

We will see below that, for each of the dynamics, the solution is symmetric,  i.e., all edges have the same capacity in equilibrium, say $z$, and hence the same resistance $1/z$. Then $\C = 3z$. Each demand is routed partly the short way and partly the long way. Since the long way has twice the resistance, the amount routed the short way is twice the amount routed the long way, i.e., $2/3$ of each demand is routed the short way and $1/3$ is routed the long way.

For each demand, let $\Delta$ be the potential drop between source and sink. 
The total energy spent is $3\Delta$. The potential drop $\Delta$ must be such that it can drive a current of $2/3$ across a wire of conductance  $z$. Thus $\Delta = 2/(3z)$. We obtain $\C + \E = 3z + 2/z$. 

\paragraph{One-Norm Dynamics:} $z$ is equal to the total current flowing across an edge and hence $z = 2/3 + 2 \cdot 1/3 = 4/3$ and $\Delta = 1/2$. So $\C = c^T z = 4$, $\E = \sum_i (\iof[i]{b})^T \iof[i]{p} = 3/2$, and $\C + \E = 11/2$. This is better than for the two-edge equilibrium.

\paragraph{Two-Norm Dynamics
}  For each edge, we have one flow of value $2/3$ and two flows of value $1/3$ and hence $z^2 = 4/9 + 2 \cdot 1/9 = 6/9$. Thus $z = \sqrt{2/3}$. $\Delta$ must be such that it can drive a current of $2/3$ across a wire of conductance $\sqrt{2/3}$ and hence $\Delta = \sqrt{2/3}$. 

Hence $\C = c^T z = 3 \cdot \sqrt{2/3} = \sqrt{6}$ and $\E = \sum_i (\iof[i]{b})^T \iof[i]{p} = 3 \cdot \sqrt{2/3} = \sqrt{6}$. Note that again we have the same value for the cost $C$ and the total energy spent $\E$. For the sum, we obtain $\C + \E = 2 \sqrt{6}$. This is better than the two-edge equilibrium.

\paragraph{Optimum:} For a general value of $z$, we have $\C + \E  = 3z + 2/z$. This is minimized for $z = \sqrt{2/3}$, i.e., the equilibrium of the two-norm 
is equal to the minimum combined cost solution.

\subsubsection{Computer Simulations}

Table~\ref{ring simulations} shows the results of a typical simulation. For the simulation we discretized the differential equation and applied an Euler forward scheme. 
\begin{table}[th]
  \centering
\begin{tabular}{|r|r|r|r|}
  \hline
  & \multicolumn{3}{|c|}{\parbox{3cm}{the final $z$-values of the three edges}}\\ \hline
two-norm dynamics   & 0.8160 & 0.8167 & 0.8166 \\ \hline
one-norm dynamics  & 1.331 & 1.327 & 1.342  \\ \hline
\end{tabular}
\caption{\label{ring simulations} The initial $z$-values were chosen randomly between $1/1000$ and $1$. 
In all cases, the system converged to the 3-edge equilibrium. Note that $0.82 \approx \sqrt{2/3}$ and $1.33 \approx 4/3$. 
}
\end{table}

\subsection{The Bow-Tie Graph}\label{Bow-Tie Graph}

\begin{figure}[t]
  \centering{\includegraphics[clip=true,trim=0 200 0 370, width=0.5\textwidth]{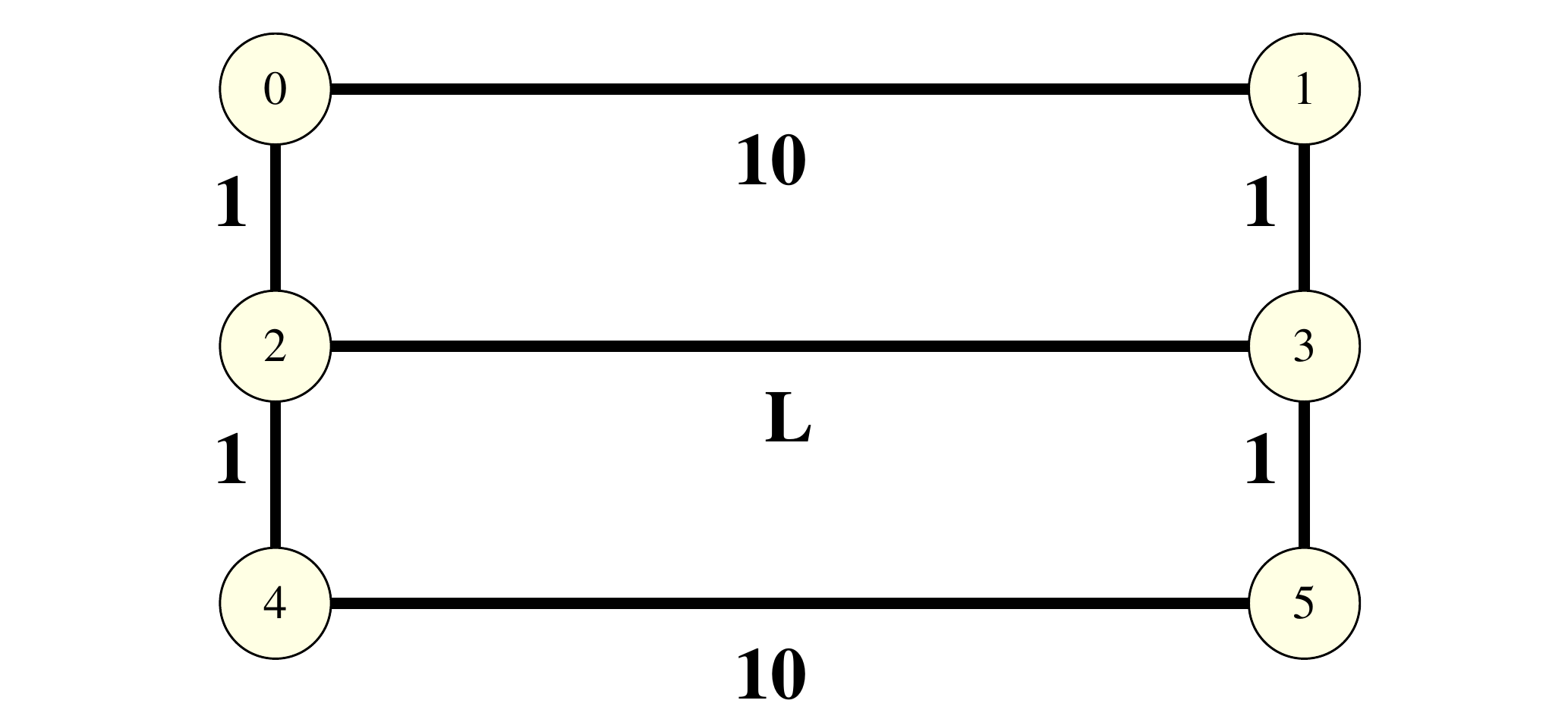}}
  \caption{\label{BowTie} The top and the bottom horizontal edge have cost 10, the middle horizontal edge has cost $L$, and all other edges have cost 1. We are sending one unit between nodes 0 and 1 and one unit between nodes 4 and 5. }
\end{figure}

      \begin{figure}[t]
        \begin{center}
         \includegraphics[width=0.9\textwidth,clip=, trim=0 100 200 0]{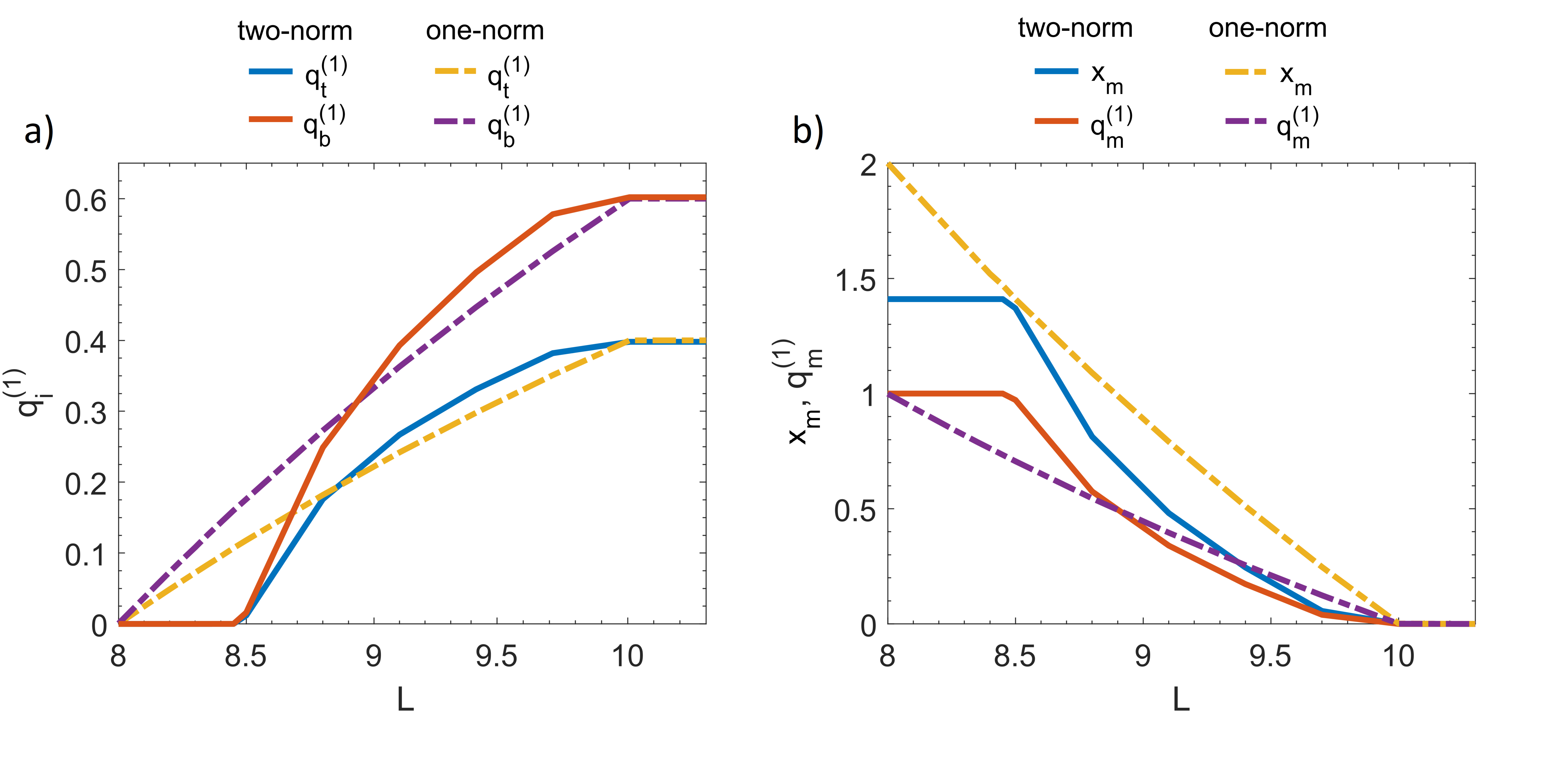}
        \end{center}
        \caption{\label{BowTie Simulation} Simulation of the Bow-Tie Graph: The plot depicts the quantities $q^{(1)}_b$, $q^{(1)}_m$, and $q^{(1)}_t$ (= the split-up of the flow from node 0 to node 1 across the three horizontal edges bottom, middle, and top) and $x_m$ (= the capacity of the middle edge) as a function of $L$
        in the range $[8, 10.3]$. For $L < 8$, the quantifies are as for $L = 8$, and for $L > 10.3$, the quantities are as for $L = 10.3$. For $q^{(2)}$ the flow across the middle edge is the same and the flow across the other edges is reversed. For all $L$, $q^{(1)}_b + q^{(1)}_m + q^{(1)}_t = 1$. 
        (1) The image on the left shows $q^{(1)}_b$ and $q^{(1)}_t$. 
        (2) The image on the right shows the capacity $x_m$ and the flow $q^{(1)}_m$ across the middle edge. We have $x_m = \sqrt{2 (q^{(1)}_m)^2} = \sqrt{2} \cdot q^{(1)}_m$ in the case of the two-norm and $x_m = 2 q^{(1)}_m$ in the case of the one-norm. }
        \end{figure}

Consider the graph shown in Figure~\ref{BowTie}; we refer to this graph as a bow-tie. The edge costs are as shown and we are sending one unit each between nodes 0 and 1 and nodes 4 and 5, i.e., $b_0 = (1,-1,0,0,0,0)$ and $b_1 = (0,0,0,0,1,-1)$. For each pair the direct path connecting the pair has length $10$, the path using the middle edge has length $L+2$ and the path using the edge connecting the other pair has length $14$. Figures~\ref{BowTie Simulation} and~\ref{BowTie Simulation Cost-Energy} show the results of a simulation. Initial $x$-values were chosen randomly in the interval $[1,10]$. We observe:
\begin{itemize}
  \item For $L \le 8$, both dynamics generate essentially the same solution. All flow is essentially routed through the middle edge. 
  \item For the two-norm dynamics: For $L < 8.5$, the sharing effect is strong and basically all flow is routed through the middle edge. Note that for $L > 8$, the path through the middle edge is not the shortest path for either demand. Starting at $L = 8.5$, the top and the bottom edge are also used. For $L \ge 10$, only the top and the bottom edge are used and this may give the impression that there is no sharing effect for large $L$. This is not the case. The solution for $L = \infty$ is easily computed analytically. Because of symmetry, a fraction $a$ of each flow is routed the short way (length 10) and a fraction $1 - a$ is routed the long way (length 14). So through each edge, we have a flow of value $a$ and a flow of value $1 - a$ and hence all edges will have the same capacity in equilibrium; call it $x$. Therefore the flows must be in the same ratio as the costs, i.e., $a/(1 - a) = 10/14$. This solves to $a = 7/12$. Then $x = \sqrt{a^2 + (1 - a)^2} = \sqrt{74}/12 \approx 8.6023$. The cost of the network is then $24 \cdot \sqrt{74}/12 = 2 \sqrt{74} \approx 17.2$ and the dissipated energy is the same. Assume now that we delete the vertical edges. Then each demand is routed separately and the bottom and the top edge will have a capacity of one each. The cost of the network will be 20 and the dissipated energy will also be 20. This is considerably more than the cost of the network constructed by our dynamics. 
  \item For the one-norm dynamics: Starting at $L = 8.05$, the top and the bottom edge are also used. For $L \ge 10.3$, only the top and the bottom edge are used.
  \item For the two-norm dynamics, the cost $\C$ and the dissipated energy $\E$ are equal in the limit; see Figure~\ref{BowTie Simulation Cost-Energy}.
    \end{itemize}

\begin{figure}
  \centering
  \begin{tabular}{|r|r|r|r|r|r|r|r|r|r|r|r|r|}
    \hline
$L$ & 6.5 & 6.8 & 7.1 & 7.4 & 7.7 & 8.0 & 8.3 & 8.6 & 8.9& 9.2 & 9.5 & 9.8  \\ \hline
$\C$ & 13.2 & 13.6 & 14.0 &  14.5  &  14.9  &  15.3  &  15.7  &  16.1  &  16.4  &  16.6  &  16.7  & 16.7   \\ \hline
    $\E$ & 13.2 & 13.6 & 14.0 &  14.5  &  14.9  &  15.3  &  15.7  &  16.1  &  16.4  &  16.6  &  16.7  & 16.7   \\ \hline
 \end{tabular}
  \caption{\label{BowTie Simulation Cost-Energy}. Simulation results for the two-norm dynamics for the bow-tie graph. The cost $\C = c^T x$ and the energy $\E = \sum_i (\iof[i]{b})^T \iof[i]{p}$ for the limit states for different values of $L$. Note that $\C = \L$ always. }
\end{figure}

      \subsection{A Case Study Inspired by~\cite{Tero-Takagi-etal}} \label{Tokyo Railroad}
  
      In~\cite{Tero-Takagi-etal} the slime molds ability to construct elegant networks in investigated. The slime is allowed to grow in a region that is shaped according to the greater Tokyo region and food is provided at many different places. Figure~\ref{Wet-Lab Experiments} shows the results of the wet-lab experiment and compares a network constructed by the slime to the railroad network around Tokyo. The paper also reports about a computer experiment. Repeatedly a pair of food sources was chosen at random and a step of the shortest path dynamics was executed. Figure 4 in~\cite{Tero-Takagi-etal} shows the results of the computer experiment. No details are given in the paper and also the positions of the food sources are not given in detail.

      We tried to repeat the experiment with the two-norm dynamics. For this purpose, we digitized the boundary of the Greater Tokyo region in the form of a polygonal region and overlayed a regular grid in which each node is connected to its up to eight neighbors (north, northwest, west, southwest, south, southeast, east, northeast) inside the region. The edge lengths are 1 for the horizontal and vertical edges and $1.41$ for the diagonal edges. We perturbed the edge lengths slightly by adding $r \cdot 0.05$ for a random integer $r \in [-3,3]$ so as to avoid many equal length path. We chose the terminals in two different ways.
      \begin{description}
      \item[First choice:] We chose the largest 25 cities cities Greater Tokyo region according to Wikipedia and generated 140 demands. Each city was connected to all other cities whose distance is below a certain threshold. For the threshold we chose about 1/2 times the diameter of the region. The left side of Figure~\ref{Tokyo Example} shows the input and Figure~\ref{Tokyo Output} shows the output of a computer simulation. 
      \item[Second choice:] We mimicked the choice of sites used in~\cite{Tero-Takagi-etal}. We generated 282 demands again between any pair of sites whose distance is below a certain threshold. The demands are 1, except if one of the terminals corresponds to Tokyo. Then the demand is seven; this is as in~\cite{Tero-Takagi-etal}. The right side of Figure~\ref{Tokyo Example} shows the input and Figure~\ref{Tero Output} shows the output of a computer simulation. 
 \end{description}

   \begin{figure}[th]
     \centerline{\includegraphics[width=0.50\textwidth,clip=true,trim=60 08 10 250]{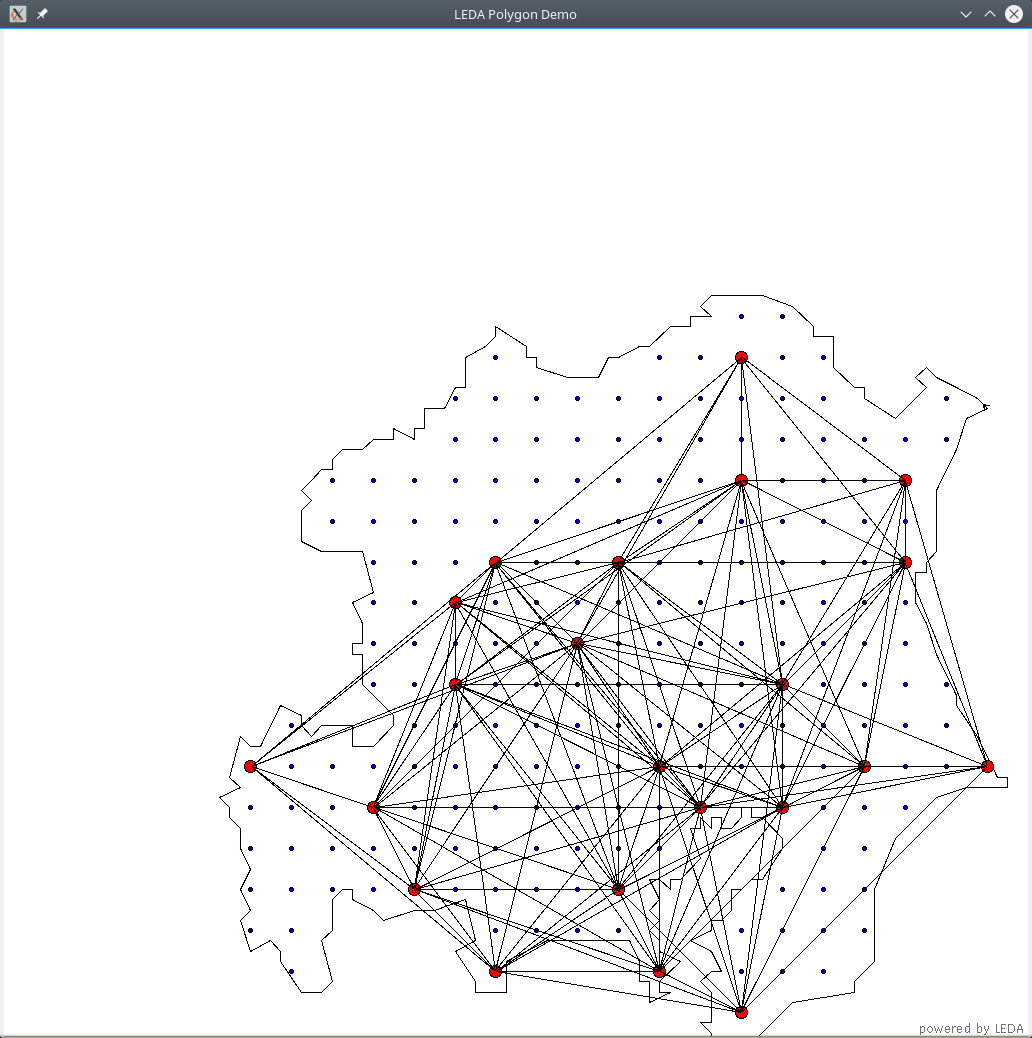}\quad\includegraphics[width=0.48\textwidth,clip=true,trim=360 70 10 380]{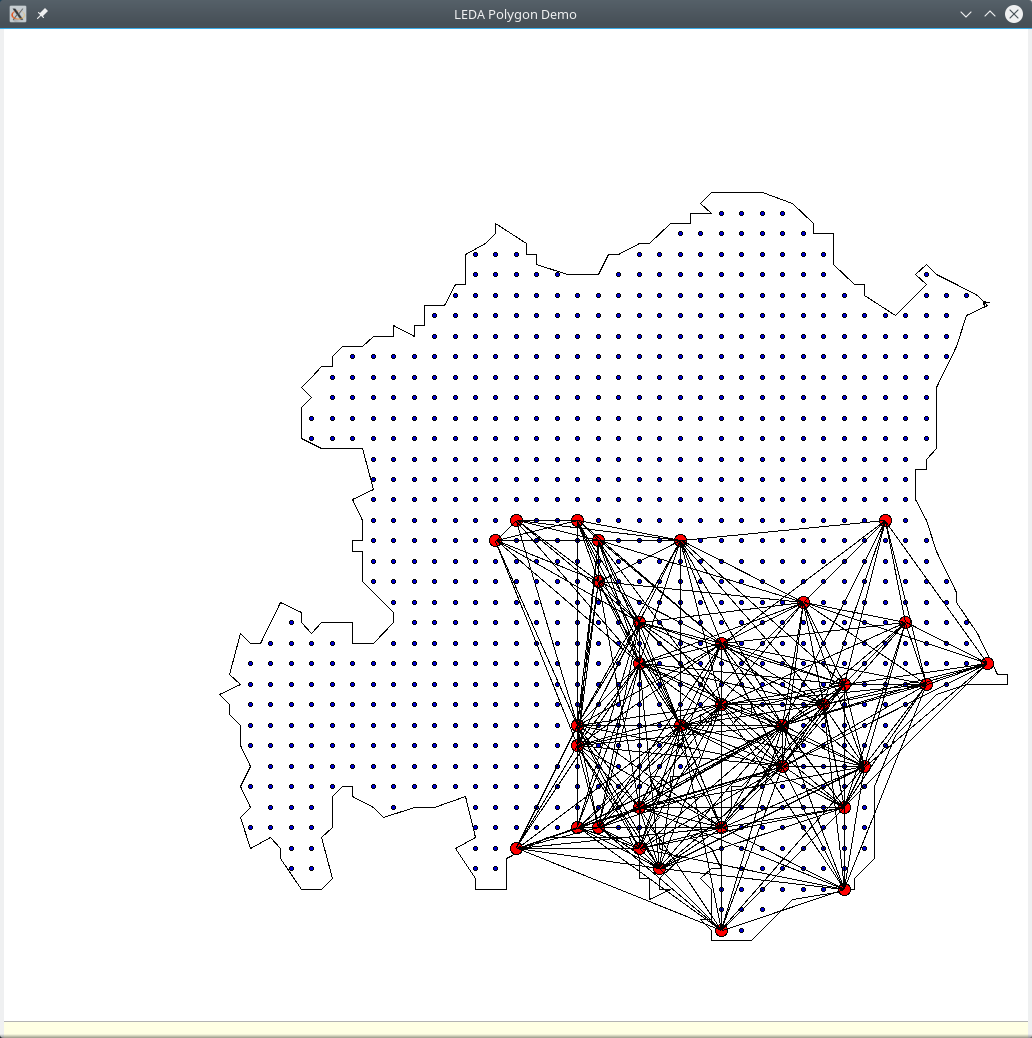}}
     \caption{The polygonal region on the left is a digitization of the Greater Tokyo Region. The red dots indicate major cities. We set up 140 demands. For each red city, we created a demand of one unit to any other red city within a certain distance threshold. The threshold is about 1/2 the distance between the topmost and the bottommost red point. The region on the right is approximately the right lower quadrant of the region on the left. For the placement of the terminals we tried to copy the placement shown in Figure~\ref{Wet-Lab Experiments}. We set up 282 demands, again between cities below a certain distance threshold. The demands are one, except if one of the terminals corresponds to Tokyo. Then the demand is seven; this is as in~\cite{Tero-Takagi-etal}.\label{Tokyo Example}}
   \end{figure}

   \begin{figure}[th]
     \centerline{\includegraphics[width=0.42\textwidth,clip=true,trim=5 30 5 50]{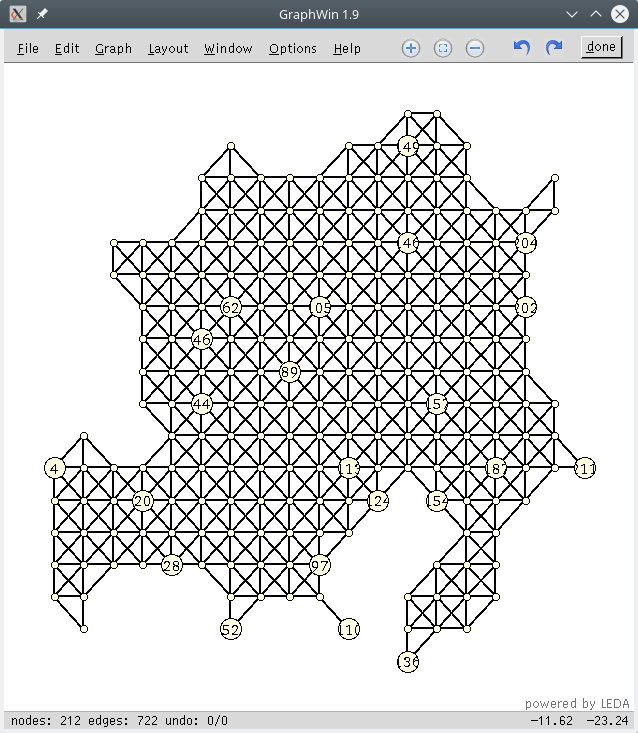}\quad \includegraphics[width=0.42\textwidth,clip=true,trim=5 30 5 50]{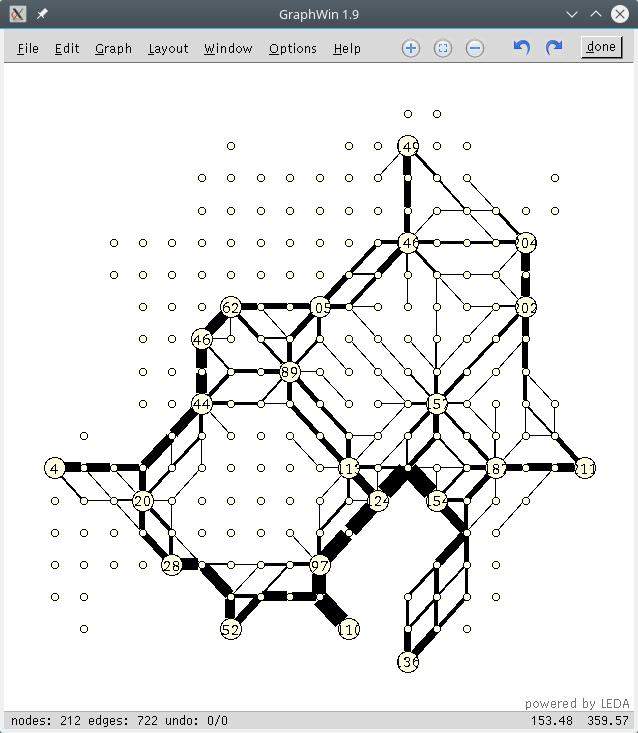}}\smallskip
    %
   
 \centerline{\includegraphics[width=0.42\textwidth,clip=true,trim=5 30 5 50]{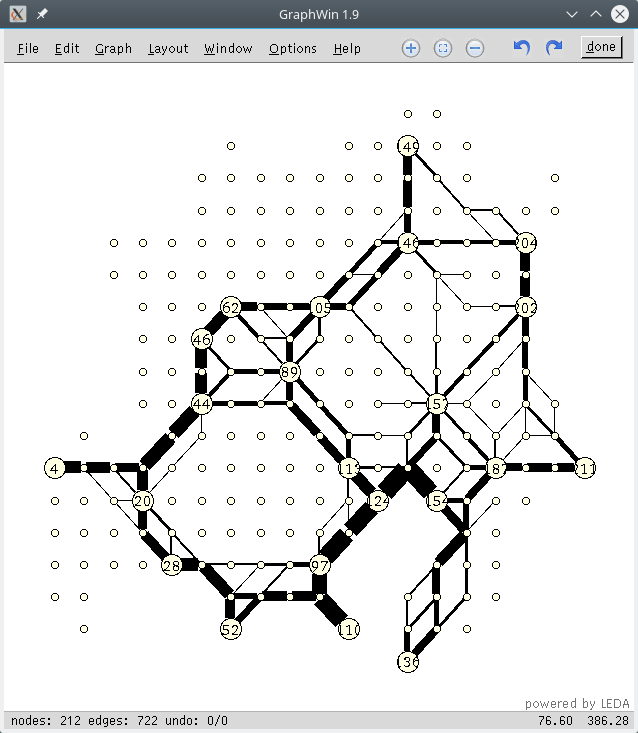}\quad \includegraphics[width=0.42\textwidth,clip=true,trim=5 30 5 50]{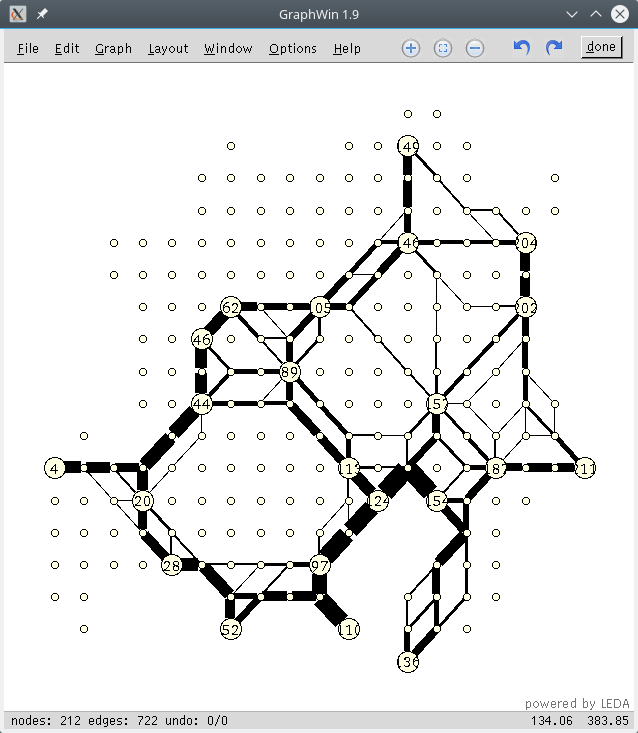}}
     
     \caption{An output of a simulation of the two-norm dynamics on the left instance in Figure~\ref{Tokyo Example}. The graph in the upper left corner shows the initial graph. Each node is connected to its up to 8 neighbors. The length of the horizontal and vertical edges is approximately 1, the length of the diagonals is approximately 1.41. All capacities are 0.5 initially and the capacity of an edge is indicated by its thickness. The following figures show the state after 1950 and 4875 iterations. For the situation after 4875 iterations, we also show the reduced graph where we iteratively removed nodes of degree one (which are not terminals). The numbers inside the nodes are unique identifiers; they have no meaning beyond this. \label{Tokyo Output}}
   \end{figure}

   \begin{figure}
     \centerline{\includegraphics[width=0.42\textwidth,clip=true,trim=5 30 5 50]{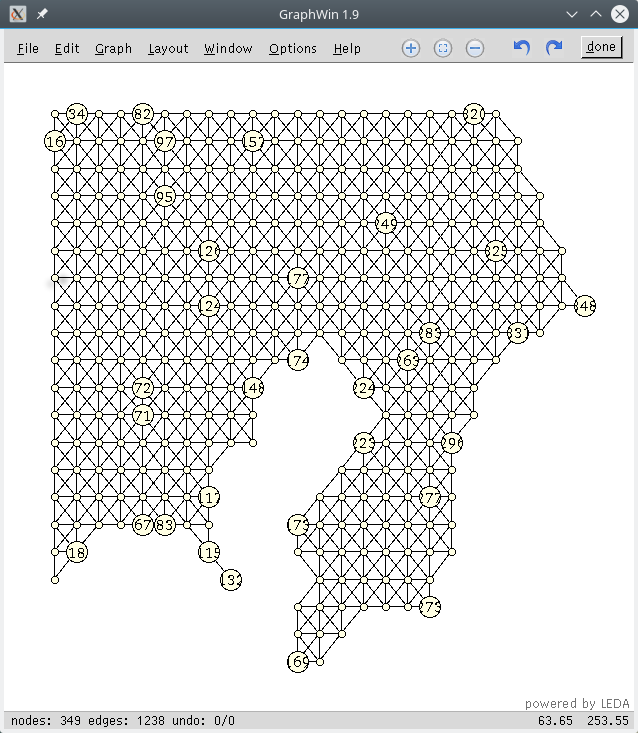}\quad \includegraphics[width=0.42\textwidth,clip=true,trim= 40 30 5 50]{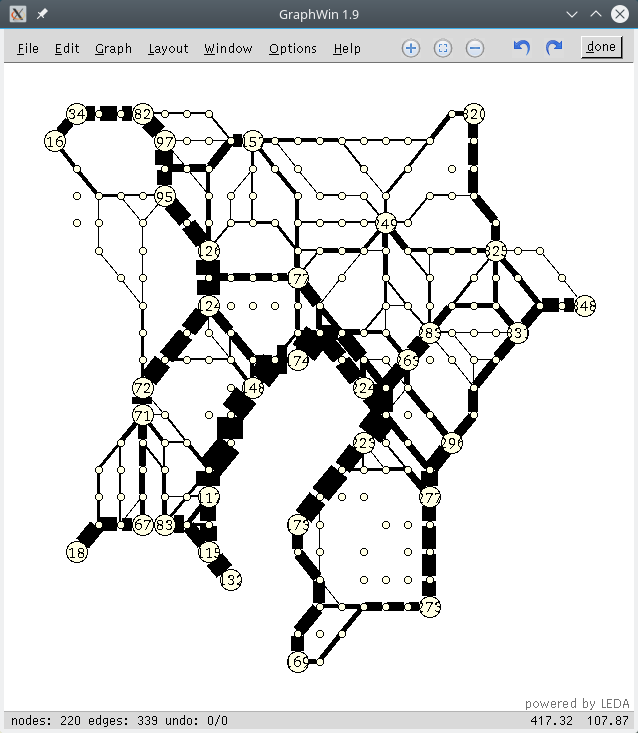}
     }

     
     
     \caption{An output of a simulation of the two-norm dynamics on the right instance in Figure~\ref{Tokyo Example}. The graph in the upper left corner shows the initial graph. Each node is connected to its up to 8 neighbors. The length of the horizontal and vertical edges is approximately 1, the length of the diagonals is approximately 1.41. All capacities are 0.5 initially and the capacity of an edge is indicated by its thickness. The figure on the right show the state after 16000 iterations where we iteratively removed nodes of degree one. \label{Tero Output}}
   \end{figure}

\section{Preliminaries}\label{Preliminaries}
We recall the definition of energy dissipation and cost. For a capacity vector $x \in \Rplus^m$
and a vector $f \in \R^m$ with $\supp(f) \subseteq \supp(x)$, we use  \[ E_x(f) = \begin{cases} \sum_e (c_e/x_e) f_e^2  &\text{if $\supp f \subseteq \supp x$},\\
    \infty &\text{if $\supp f \not\subseteq \supp x$.} \end{cases}\]  to denote the \emph{energy dissipation} of $f$ with respect to $x$. Strictly speaking we should sum only over the $e$ in $\supp x$. We use the convention $0^2/0 = 0$ to justify summing over all edges $e$. Further, we use
  \[ \C(f)= \sum_e c_e \abs{f_e} = c^T \abs{f}\]  to denote the \emph{cost} of $f$. Note that
  \[ E_x(x) = \sum_e (c_e/x_e) x_e^2 = \sum_e c_e x_e = \C(x). \]

We use $R$ to denote the diagonal matrix with entries $c_e/x_e$. Energy-minimizing solutions are induced by node potentials $p \in \R^n$ according to the following equations:
\begin{align}
  b &= Aq \label{feasibility},\\
  q &= R^{-1} A^T p  \label{definition of q},\\
  AR^{-1} A^T p &=b.    \label{definition of p}
\end{align}
We give a short justification why the equations above characterize the energy minimizing solution to the linear system. The energy minimizing solution $q$ minimizes the quadratic function $\sum_e (c_e/x_e) q_e^2$ subject to the constraints $Aq = b$ and $\supp(q) \subseteq \supp(x)$.  The KKT conditions (see~\cite[Subsection 5.5]{Boyd-Vandenberghe}) state that at the optimum, the gradient of the objective is a linear combination of the gradients of the constraints, i.e., 
\[       2 (c_e/x_e) q_e = \sum_i \iof[i]{p} A_{i,e}  \quad \text{for all $e \in supp(x)$}\]
for some vector $p \in \R^n$ and $q_e = 0$ for $e \not\in \supp(x)$. Absorbing the factor $2$ into $p$ yields equation \lref{definition of q}. Substitution of \lref{definition of q} into \lref{feasibility} gives \lref{definition of p}. The energy-minimizing solution is unique. It exists if and only if $b \in \Im A$. Node potentials $p$ are not unique, but the values of $b^Tp$ and $p^T L(x)p$ are equal fo all solutions of~\lref{definition of p}.

\begin{lemma}\label{Basics} Assume $x > 0$. Then $\Ker L(x) = \Ker A^T$ and $\Im L(x) = \Im A$. The values $b^T p$, $p^T L(x) p$ and $q = X C^{-1} A^T p$ do not depend on the particular solution of $L(x)p = b$. \end{lemma}
\begin{proof} Clearly, $\Ker A^T \subseteq \Ker L(x)$. So assume $z \in \Ker L(x)$. Then $L(x) z = 0$ and hence $z^T L(x) z = 0$. Let $D^{1/2}$ be the diagonal matrix with entries $\sqrt{x_e/c_e}$. Then
  \[ 0 = z^T L(x) z = z^T A D^{1/2} D^{1/2} A^T z = \normtwo{D^{1/2} A^T z}^2\]
  and hence $D^{1/2} A^T z = 0$ and further $0 = A^T z$. So $z \in \Ker A^T$.

  Clearly, $\Im L(x) \subseteq \Im A$. So assume $b \not\in \Im L(x)$. Then the rank of the matrix obtained by augmenting $L(x)$  by the column $b$ is larger than the rank of $L(x)$ (Rouch\'{e}-Capelli theorem) and hence there is a vector $r$ such that $r^T b \not= 0$ and $r^T L(x) = 0$. Since $L(x)$ is symmetric, $L(x) r = 0$ and hence $r \in \Ker L(x) = \Ker A^T$. So $0 = A^T r = (r^T A)^T$. Thus $r$ also proves $b \not\in \Im A$. 

Let $p$ and $\bar{p}$ be node potentials. Then $L(x) p = b = L(x) \bar{p}$ and hence $\bar{p} - p \in \Ker L(x)$. Then 
  \[ b^T \bar{p} = b^T p + b^T (\bar{p} - p) = b^T p + p^T L(x)^T (\bar{p} - p) = b^T p + p^T L(x) (\bar{p} - p) = b^T p\]
  and
  \[ X C^{-1} A^T \bar{p} =  X C^{-1} A^T p + X C^{-1} A^T (\bar{p} - p) = X C^{-1} A^T p.\]
  Finally, $b^T p = p^T L(x) p$. \end{proof}

For the arc-node incidence matrix $A$ of a connected graph, the kernel $\Ker A^T$ consists of the all-ones vector in $\R^n$. We can make the node potential unique by requiring $p_v = 0$ for some fixed node $v$, i.e., by grounding node $v$.

\begin{lemma} Let $\ell$ be the dimension of $\Ker A^T$ and let $K \in \R^{n\times \ell}$ be a matrix whose columns form a basis of $\Ker A^T$. Let $V' \subseteq [n]$ with $\abs{V'} = \ell$ be such that the submatrix of $K$ with rows selected by $V'$ is nonsingular. Then the solution $p$ to $L(x) p = b$ with $p_v = 0$ for all $v \in V'$ is unique, i.e. ``grounding all nodes in $V'$ makes the potential unique''. 
\end{lemma}
\begin{proof} Observe first that such a solution exists. Let $p$ be an arbitrary solution to $L(x) p = b$. Then there is a vector $\lambda \in \R^\ell$ such that $(K \lambda)_v = p_v$ for all $v \in V'$ and hence $p - K \lambda$ is the desired node potential. Assume now that we have two solutions $p$ and $p'$ with $p_v = p_v'$ for all $v \in V'$. Then $p - p' \in \Ker L(x) = \Ker A^T$ and $(p - p')_v = 0$ for all $v \in V'$. Since $p - p' \in \Ker A^T$ there is a $\lambda \in \R^\ell$ such that $p - p' = K \lambda$. Then $(K \lambda)_v = 0$ for all $v \in V'$. Since the columns of $K$ are independent, this implies $\lambda = 0$ and hence $p = p'$.
\end{proof}

The next Lemma gives alternative expressions for the energy $E_x(q)$ of the minimum energy solution. 

\begin{lemma}$E_x(q) = \sum_e (c_e/x_e) q_e^2 = b^T p = p^T L(x) p$, where $p$ is any solution of~\lref{definition of p}. \end{lemma}
\begin{proof}
  This holds since \[E_x(q) = q^T R q = p^T A R^{-1} R R^{-1} A^T p = p^T A R^{-1} A^T p = p^T L(x) p = p^T b.\]
  \proofendswithequation
 \end{proof}

Finally, we recapitulate a bound on the components of $q$ established in~\cite{SV-LP} and slightly improved form in~\cite[Lemma 3.3]{BeckerBonifaciKarrenbauerKolevMehlhorn}.

\begin{lemma} Let $D$ be the maximum absolute value of a square submatrix of $A$. Then 
  $\abs{q_e} \le D\onenorm{b}$ for every $e \in [m]$. \end{lemma}

\section{Existence of a Solution}\label{Existence}

We prove Theorem~\ref{thm: Existence}. The right-hand side~\lref{eq:GenDyn} is locally Lipschitz-continuous in $x$.  The function $g_e$ is locally Lipschitz by assumption, the $\qi$'s are infinitely often differentiable rational functions in the $x_e$ and hence locally Lipschitz. Furthermore, locally Lipschitz-continuous functions are closed under additions and multiplications.
Thus $x(t)$ is defined and unique for $t \in [0,t_0)$ for some $t_0$.

Since $g_e$ is non-negative, we have $\dot{x}_e \ge -x$ and thus $x_e \ge x_e(0) e^{-t}$. Hence, $x(t) > 0$ for all $t$. By assumption $\iof[i]{b} \in \Im A$ for all $i$, and hence whenever $x(t) > 0$, we have solutions $\qi$ with $\supp(\qi) \subseteq \supp(x)$. 

In  Section~\ref{Lyapunov Function}, we will show that $\L$ is a Lyapunov function for the dynamics~\lref{eq:GenDyn}. Thus
\[  c^Tx \le \L(x) \le \L(x(0)) \]
and hence $x$ stays in a bounded domain.

It now follows from general results about the solutions of ordinary differential equations~\cite[Corollary 3.2]{Hartman} that $t_0 = \infty$.

\section{Fixed Points}\label{Fixed Points}

A point $x$ is a fixed point iff $\dot{x} = 0$. We use $\F_g$ for the set of fixed points of~\lref{eq:GenDyn}.

   \begin{lemma}[The fixed points of the generalized Physarum dynamics]\label{lem:GenPhyDin}
   	 $x \in \F_g$ iff for all
     $e$ either $x_e = 0$ or $\twonorm{\row[e]{\Lambda}}= 1$. The latter condition is equivalent to $x_e = \twonorm{Q_e}$ or $\twonorm{\row[e]{(A^T)} P}= c_e$. For $x \in \F_g$, $\C(x) = \E(x)$. 
   \end{lemma}
   \begin{proof} We have $\dot{x} = 0$ iff we have $x_e = 0$ or $g_e(\twonorm{\Lambda_e}) = 1$ for all $e$. Since $g_e$ is increasing and $g_e(1) = 1$, the latter condition is tantamount to $\twonorm{\Lambda_e} = 1$ which expands to $\sum_i (\row[e]{(A^T)} \col[i]{P})^2 = c_e^2$. Multiplying both sides by $(x_e/c_e)^2$ yields 
$x_e^2 = \sum_i (\entry[e,i]{Q})^2$. 

For $x \in \F_g$, we have 
     \[ \E(x) = \sum_{e} \sum_i \frac{c_e}{x_e} (\entry[e,i]{Q})^2 = \sum_{e} \frac{c_e}{x_e} \cdot x_e^2 = \sum_e c_e x_e=  \C(x).\]
     \end{proof}

\section{Lyapunov Function}\label{Lyapunov Function}
Let 
\[
\L(x )=\frac{1}{2}\left(c^{T}x +\sum_{i=1}^{k} (\iof[i]{b})^{T} \iof[i]{p} \right).
\]
We will show that $\L$ is a Lyapunov function for the dynamics~\lref{eq:GenDyn}. The function $\L$ was introduced in~\cite{Facca-Daneri-Cardin-Putti}. For $k = 1$, \cite{Facca-Cardin-Putti} shows that $\L$ is a Lyapunov function for the one-norm dynamics and~\cite{Karrenbauer-Kolev-Mehlhorn:NonUniformPhysarum} shows that this holds true also for the generalized Physarum dynamics. The calculations below generalize the calculations in these papers. They are similar to the calculations in~\cite[Lemma 2.6]{B19}.

\begin{lemma}[Gradient of $\L$]\label{Gradient of phi} For all $e \in E$,
  \begin{equation}
    \frac{\partial}{\partial x_{e}}\L(x) = \frac{c_{e}}{2}( 1- \twonorm{\row[e]{\Lambda}}^2).\label{eq:gradPhi}
  \end{equation} \end{lemma}
\begin{proof} Recall $L(x) = A X C^{-1} A^T$. Let $e \in [m]$ be arbitrary. Then $\frac{\partial}{\partial x_e} L(x) = \frac{1}{c_e} \col[e]{A} \row[e]{(A^T)}$. From $L(x) p = b$ and $\frac{\partial}{\partial x_e} b = 0$,  we obtain
\[  0 = \frac{\partial}{\partial x_e} L(x) p =  \frac{\partial L(x)}{\partial x_e} p + L(x) \frac{\partial p}{\partial x_e}\]
and thus
\[   L(x) \frac{\partial p}{\partial x_e} = - \frac{1}{c_e} \col[e]{A} \row[e]{(A^T)} p.\]
Hence, we have
\[ \frac{\partial}{\partial x_{e}}b^{T}p=b^{T}\frac{\partial p}{\partial x_{e}}=p^{T}L(x)\frac{\partial p}{\partial x_{e}}=-\frac{1}{c_{e}}p^{T} \col[e]{A} \row[e]{(A^T)}p=-c_{e}\left(\frac{\row[e]{(A^T)}p}{c_{e}}\right)^{2},\]
and more generally, 
\[
\frac{\partial}{\partial x_{e}}\sum_{i}(\iof[i]{b})^{T}\iof[i]{p}=-c_{e}\sum_{i}\left(\frac{\row[e]{(A^T)}\iof[i]{p}}{c_{e}}\right)^{2}=-c_{e}\twonorm{\row[e]{\Lambda}}^{2}.
\]
The claim follows. 
\end{proof}

\begin{theorem}\label{lem:GenEqSet} 
	The function $\L: \Omega \mapsto\mathbb{R}$
	is a Lyapunov function for the dynamics (\ref{eq:GenDyn}), i.e., $\L(x) \ge 0$ for all $x \in \Omega$, and $\frac{d}{dt} \L(x(t)) \le 0$ for all $t$. Let
        \[      \V = \set{x \in \Omega}{\langle \nabla \L(x), \dot{x} \rangle = 0}. \]
        Then $\V = \mathcal{F}_g$.
\end{theorem}
\begin{proof} $\L(x) \ge 0$ for all $x \in \Omega$ is obvious. 

  Since $\frac{d}{dt} \L(x(t)) = \langle \nabla \L(x), \dot{x} \rangle$, we obtain
  \[ \frac{d}{dt} \L(x(t))  = \sum_e \frac{c_e}{2} (1 - \twonorm{\row[e]{\Lambda}}^2)\cdot x_e (g_e(\twonorm{\row[e]{\Lambda}}) - 1) \le 0,\]
where the inequality holds since $g_e(\twonorm{\row[e]{\Lambda}}) - 1$ and $\twonorm{\row[e]{\Lambda}} - 1$ have the same sign, as $g_e$ is a non-negative and increasing function with $g_e(1) = 1$.

We have equality if and only if for all $e$ either $x_e = 0$ or $\row[e]{\Lambda} = 1$. Thus $x \in \V$ if and only if $x \in \mathcal{F}_g$. 
\end{proof}






      \section{Further Properties of the Lyapunov Minimum}\label{Further Properties}

We give two alternative characterizations for the minimum of the Lyapunov function. This extends~\cite[Proposition 2]{Facca-Cardin-Putti} from $k = 1$ to arbitrary $k$. 
                      
\begin{theorem}\label{MQ = MP = ML} The following quantities $\MQ$, $\MP$, and $\ML$ are equal. 
\begin{align}
  \MQ &= \min_{Q \in \R^{m \times k}} \set{\sum_e c_e  \normtwo{Q_e} } { A Q = B }  \label{minQ},\\
  \MP &= \max_{P \in \R^{n \times k}}\set{ \Tr[B^TP] }{ \normtwo {\row[e]{(A^T)} P} \le c_e \text{ for all $e$}}  \label{maxP},\\
  \ML &= \min_{x \in \R^m_{\ge 0}}  \L(x). \label{minL}
  \end{align}
  Moreover, there are optimizers $Q^*$, $P^*$ and $x^*$ such that
  \begin{align*}  x^*_e &= \normtwo{\row[e]{Q^*}} \quad\text{ for all $e$},\\
    L(x^*) P^* &= B ,\\
    Q^* &= X^* C^{-1} A^T P^*.
  \end{align*}
\end{theorem}

\begin{lemma}\label{lem:KKT} 
	Let $Q^*$ be a minimizer of~\lref{minQ} and let $x^*$ be defined by $x_e^*= \normtwo{\row[e]{Q^*}}$ for all $e$. Then $x^* \in \mathcal{F}_g$. Moreover, there is a potential matrix $P \in \R^{n \times k}$ such that $L(x^*) P = B$ and $Q = X^* C^{-1} A^T P$, $\sum_e c_e \normtwo{\row[e]{Q}} = \Tr[B^TP] = \L(x^*)$, and  $\normtwo{\row[e]{(A^T)} P} \le c_e$ for all $e$. The objective values of~\lref{minQ} to~\lref{minL} satisfy $\ML \le \MQ \le \MP$. 
\end{lemma}

\begin{proof} We start by slightly reformulating the minimization problem~\lref{minQ}. This is necessary since the function $\row[e]{Q} \mapsto \normtwo{\row[e]{Q}}$ is not differentiable for $\row[e]{Q} = 0$ and hence the KKT-conditions cannot be applied. We formulate equivalently:
\[ 
  \min \sum_e c_e x_e \text{ subject to } AQ = B,\ x_e^2 \ge \normtwo{\row[e]{Q}}^2,\ x_e \ge 0 \text{ for all $e$},
\]
with variables $Q \in \R^{m \times k}$ and $x \in \R^m$. 
%
%
Let $Q^*$ and $x^*$ be an optimal solution. Then clearly $x^*_e = \normtwo{\row[e]{Q^*}}$ for all $e$. 
Using the Lagrange multipliers $P \in \R^{n\times k}$ for the equations $AQ = B$, and $\alpha \in \R_{\ge 0}^m$ and $\beta \in \R_{\ge 0}^m$ for the inequalities, the KKT conditions~\cite[Subsection 5.5]{Boyd-Vandenberghe} become
\begin{align} c_e - 2 \alpha_e x^*_e - \beta_e &= 0 \quad\text{for all $e$}\label{condition 1} ,\\
  P^T  \row[e]{(A^T)} + 2 \alpha_e \row[e]{Q^*} &= 0 \quad\text{for all $e$}\label{condition 2} ,\\
  \alpha_e ((x_e^*)^2 - \normtwo{\row[e]{Q^*}}^2 )&= 0 \quad\text{for all $e$}\label{condition 4},\\
  \beta_e x_e^* & = 0 \quad \text{for all $e$}.  \label{condition 5}
\end{align}
Here the first two conditions state that at the optimum, the gradient of the objective with respect to the variables $x_e$ and $\entry[e,i]{Q}$ must be linear combinations of the gradients of the active constraints and the last two conditions are complementary slackness (= a Lagrange multiplier can only be non-zero if the constraint is tight). We also have the feasibility constraints
\begin{align}
        A Q^*  &= B \label{condition 3},\\
  x_e^*  & \ge 0 \text{ and } x_e^* \ge \normtwo{\row[e]{Q}} \quad\text{for all $e$.}
\end{align}

Separating the two terms in~\lref{condition 2}, squaring and summing over $i$, and using~\lref{condition 2} and~\lref{condition 1}, we obtain 
\[   \normtwo{\row[e]{(A^T)} P}^2 =  \sum_i  ((\iof[i]{p})^T \entry[e,i]{A})^2 = 4 \alpha_e^2 \normtwo{Q^*_e}^2 = 4\alpha_e^2(x^*_e)^2 = (c_e - \beta_e)^2 \le c_e^2,\]
where the last inequality uses $\beta_e = 0$ if $x^*_e > 0$ by~\lref{condition 5} and $\beta_e = c_e$ if $x^*_e = 0$ by~\lref{condition 1}.

If $\row[e]{Q^*} \not= 0$, then $x^*_e \not= 0$ and hence $\beta_e = 0$ and $c_e = 2 \alpha_e x^*_e$ or $2 \alpha_e = c_e/x^*_e$. In particular, $\alpha_e \not= 0$ and hence~\lref{condition 2} implies
\begin{equation}\label{reproduces Q}           \row[e]{Q^*} = \frac{1}{2 \alpha_e} \sum_v \row[v]{P} A_{v,e} = \frac{x^*_e}{c_e} \row[e]{A}^T P. \end{equation}
This equation also holds if $\row[e]{Q^*} = 0$ and hence $x^*_e = 0$. 
Multiplying by $\row[e]{(A^T)}$ from the left and summing over $e$ yields
\begin{equation}\label{P is potential}
  B =  A Q^*= A X C^{-1} A^T P.\end{equation}
Thus $\col[i]{P}$ is a potential for the $i$-th problem with respect to $x^*$ and, by~\lref{reproduces Q} $\col[i]{(Q^*)}$ is the corresponding electrical flow. Thus $x^* \in \F_g$ by Lemma~\ref{lem:FixPoints}. Moreover, 
\begin{align*}
\sum_i (\col[i]{P})^T \iof[i]{b} &= \sum_i (\col[i]{P})^T A \col[i]{(Q^*)}\\
& = \sum_{i,v,e} P_{v,i} A_{v,e} Q^*_{e,i} \\
&= \sum_{i,v,e,\ \row[e]{Q} \not= 0} \entry[v,i]{P} \entry[v,e]{A} \entry[e,i]{Q^*} \\
&= \sum_{i,e,\ \row[e]{Q} \not= 0} c_e \frac{\entry[e,i]{Q^*} \entry[e,i]{Q^*}}{\normtwo{\row[e]{Q^*}}} \\
&= \sum_e c_e \normtwo{\row[e]{Q^*}}.
\end{align*}
Here the fourth equality comes from~\lref{reproduces Q} and $x^*_e \not= 0$ if $\row[e]{Q^*}\not= 0$; note that $$\sum_v A_{ve}P_{vi} = \row[e]{(A^T)} P_i = \frac{c_e}{x_e^*} \entry[e,i]{Q^*} = c_e \frac{\entry[e,i]{Q^*}}{\normtwo{\row[e]{Q^*}}}.$$
We conclude that $P$ is a feasible solution to~\lref{maxP}. Thus $\MP \ge \MQ$. 

  Since $x^* \in \F_g$, $\L(x^*) = c^T x^* = \Tr[P^T L(x^*) P]$. Also,
  $x^*_e = \normtwo{\row[e]{Q^*}}$ by definition of $x^*$.  Thus
  \[   \sum_e c_e \normtwo{\row[e]{Q^*}} = c^T x^* = \L(x^*)\]
  and hence $\ML \le \MQ$. 
\end{proof}

\begin{lemma} $\MP \le \ML$. 

\end{lemma}
\begin{proof} The constraint $\normtwo{\row[e]{(A^T)} P} \le c_e$ in~\lref{maxP} can be equivalently written as
\[ \frac{c_e}{2}\left(\normtwo{\frac{1}{c_e}\row[e]{(A^T)} P}^2 - 1\right) \le 0.\]
Then the Lagrange dual with non-negative multipliers $x_e$ is an upper bound for $\MP$, i.e, 
\[  
\MP \le \inf_{x \ge 0} \sup_{P} \sum_i (\iof[i]{b})^T \col[i]{P} - \sum_e  \frac{x_e c_e}{2}\left(\normtwo{\frac{1}{c_e} \row[e]{(A^T)} P}^2 - 1\right).
\]
The inner supremum can be reformulated as
\begin{equation}\label{inner sup} 
\sup_{P} \sum_i (\iof[i]{b})^T\col[i]{P} - \frac{1}{2}\sum_i (\col[i]{P})^{T}L(x)\col[i]{P} + \frac{1}{2} c^T x, 
\end{equation}
since $(x_e/c_e) \sum_i (\row[e]{(A^T)} \col[i]{P})^2 = \sum_i (\col[i]{P})^T \col[e]{A} (x_e/c_e) \row[e]{(A^T)} \col[i]{P}$. 
Only the first two terms in~\lref{inner sup} depend on $P$. We want to determine the maximizer\footnote{In the proof of Lemma~\ref{Basics}, we have seen that $L(x) = A D^{1/2} D^{1/2} A^T$ and hence
$b_i^T \col[i]{P} - (\col[i]{P})^T L[x] \col[i]{P} = b_i^T \col[i]{P} - \twonorm{D^{1/2} A^T \col[i]{P}}^2$. Thus the maximizer is a finite point.}$P(x)$. Taking partial derivatives 
with respect to the vectors $\col[i]{P}$ leads to the system 
\[    A X C^{-1} A\cdot \col[i]{P}(x) = \iof[i]{b} \quad \text{for all $i$,} \]
i.e. $\col[i]{P}(x)$ is a solution to $L(x) \col[i]{P}(x) = \iof[i]{b}$ for each $i$. 
Since
\[
\sum_{i}(\iof[i]{b})^{T}\col[i]{P}(x)=\Tr[B^{T}P(x)]=\Tr[P(x)^{T}L(x)P(x)]=\sum_i (\col[i]{P}(x))^{T}L(x)\col[i]{P}(x)
\]
substituting into~\lref{inner sup} yields
\[
\sup_{P}\sum_{i}(\iof[i]{b})^{T}\col[i]{P}-\frac{1}{2}\sum_i (\col[i]{P})^{T}L(x)\col[i]{P}+\frac{1}{2}c^{T}x=\frac{1}{2}\left(\Tr[B^{T}P(x)]+c^{T}x\right)=\L(x).
\]
\end{proof}

\begin{lemma} Let $x^* \in \R^m_{\ge 0}$ be a minimizer of $\L(x)$. Then $x^* \in \F_g$. Let $P$ be a solution to $L(x^*)P = B$ and let $Q = X^* C^{-1} A^T P$. Then $\sum_e c_e \normtwo{\row[e]{Q}} = \L(x^*)$ and hence $\MQ \le \ML$.

\end{lemma}
\begin{proof} Since $\L(x(t))$ is a Lyapunov function of the generalized Physarum dynamics we have $x^* \in \V$. Since $\V = \F_g$, $x^*$ is a fixed point and hence for all $e$, either $x^*_e = 0$ or $\twonorm{\row[e]{\Lambda}} = 1$. Since $x^*$ is a fixed point, we have $\L(x^*) = c^T x^* = \Tr[P^T L(x^*) P]$ and $x^*_e = \normtwo{\row[e]{Q}}$ for all $e$. Thus
  \[   \sum_e c_e \normtwo{\row[e]{Q}} = c^T x^* = \L(x^*) \]
  and hence $\MQ \le \ML$. 

  \end{proof}

\section{Convergence to the Lyapunov Minimizer}

      We show that the dynamics converges to the minimizer $x^*$ of the Lyapunov function under the assumption that the set of fixed points of the dynamics is a discrete set.

      \begin{assumption}[Discrete Set of Fixed Points] $\F_g$ is a finite set of points. For any two points in $\F_g$, the values of $\L$ are distinct.\label{Discrete Set of Fixed Points} \end{assumption}

      \begin{theorem}\label{thm:GenPhyDyn_conv_to_OPT} Let $x^* = \argmin_{x \ge 0} \L(x)$. Under the additional assumption~\ref{Discrete Set of Fixed Points}, the generalized Physarum dynamics $x(t)$ converges to $x^*$. 
      \end{theorem}
      \begin{proof} Since $\L(x(t))$ is non-increasing and non-negative, the dynamics $x(t)$ converges to the set
      $\V$. By Theorem~\ref{lem:GenEqSet}, $\V = \F_g$. Since $\F_g$ is assumed to be a finite set and any two fixed points have distinct values of $\L$, there is a fixed point $\hat{x} = \lim_{t \rightarrow \infty} x(t)$. Assume for the sake of a contradiction, $\L(\hat{x}) > \L(x^*)$. Let $P(t)$ be the node potential corresponding to $x(t)$ and let $\hat{P}$ be the potential corresponding to $\hat{x}$; recall that node potentials are unique. Since $P(t)$ is a continuous function of $x(t)$, $P(t) \rightarrow \hat{P}$ as $t \rightarrow \infty$.
Let $\hat{E} = \set{e}{\normtwo{\row[e]{A^T} \hat{P}} \le c_e} \subseteq E$ and consider the following chain of inequalities:
		\begin{align*}
                  \max_{P}\set{\Tr[B^{T}P]}{\normtwo{A_{e}P}\le c_{e}\ \text{for all }e\in\hat{E}} &\geq \Tr[B^{T}\hat{{P}]}\\ &=\mathcal{L}(\hat{x})\\ &>\mathcal{L}(x^{\star})\\
                  &=\max_{P}\set{\Tr[B^{T}P]}{\normtwo{A_{e}P}\le c_{e}\ \text{for all }e\in E},
		\end{align*}
        where the first inequality follows by the definition of $\hat{E}$,
        the first equality follows from Lemma~\ref{lem:FixPoints},
        the strict inequality holds by assumption and the last equality 
        follows from Theorem~\ref{MQ = MP = ML}. We conclude that $\hat{E}$ is a proper subset of $E$. 

        Let $e\in E\backslash\hat{E}$ be arbitrary. Then
		$\normtwo{\row[e]{(A^T)} \hat{P}} > c_e$ and hence there are $t_{0}>0$ and $\epsilon>0$ such that for every $t\ge t_{0}$
		we have 
		\[
			\lVert\Lambda_{e}(t)\rVert_{2}=\frac{\lVert \row[e]{(A^T)}P(t)\rVert_{2}}{c_{e}}>1+\epsilon.
		\]
		Since $g_{e}$ is an increasing function with $g_{e}(1)=1$, 
		there is an $\alpha>0$ such that for all $t \ge t_0$
		\[
			g_{e}\left(\lVert\Lambda_{e}(t)\rVert_{2}\right)\geq g_{e}\left(1+\epsilon\right)=1+\alpha.
		\]
		Then, for the generalized dynamics we have
		\[ 
			\dot{x}_{e}(t)=x_{e}(t)\cdot(g_{e}(\lVert\Lambda_{e}(t)\rVert_{2})-1)\ge x_{e}(t)\cdot(g_{e}(1+\epsilon)-1)\ge\alpha x_{e}(t).
		\]
		Further, by Gronwall's Lemma, it follows that
		\[
			x_{e}(t)\geq x_{e}(t_{0})\cdot e^{\alpha t},
		\]
		and thus
		\[
			\hat{x}_{e}=\lim_{t\rightarrow\infty}x_{e}(t)\geq x_{e}(t_{0})\cdot\lim_{t\rightarrow\infty}e^{\alpha t}=+\infty.
		\]
		This is a contradiction to the fact that $\hat{x}_{e}$ is bounded.

                Finally, if $\L(x(t))$ converges to $\min_{x \ge 0} \L(x)$ and the minimizer $x^*$ of $\L$ is unique, then $x(t)$ must converge to $x^*$. 
              \end{proof}

              We conjecture that $x(t)$ always converges to some minimizer of $\L$. If there are several minimizers of $\L$, the limit depends on the initial configuration and the function $g$. Consider the following simple example. We have a network with two nodes connected by two links of the same cost, $k = 1$ and the goal is to send one unit between the two nodes. Let $x_1$ and $x_2$ be the capacities of the two links, respectively. For $g(z) = z$, any combination $(x_1,x_2)$ with $x_1 + x_2 = 1$ is a fixed point. 

\section{A Connection to Mirror Descent}\label{sec:MirrorDescent}

We show that the \textit{mirror descent} dynamics on the Lyapunov function $\L$
is equal to a variant of the non-uniform squared Physarum dynamics. 

\begin{lemma}
	The dynamics 
	\[
		\frac{d}{dt}x_{e}(t)=
		\frac{c_{e}}{2} x_{e}(t) \left( \twonorm{\row[e]{\Lambda}}^2 -1 \right)
	\]
	is equivalent to the \emph{mirror descent} dynamics on the Lyapunov function $\L$.
\end{lemma}
\begin{proof}
By Lemma~\ref{Gradient of phi}, we have for every index $e\in E$ that
	\begin{equation}\label{eq:defPhi}
		\frac{\partial}{\partial x_{e}}\L(x)= \frac{c_{e}}{2}( 1- \twonorm{\row[e]{\Lambda}}^2).
	\end{equation}
	On the other hand, the mirror descent dynamics on the Lyapunov function $\L$ is given by
	\begin{equation*}\label{eq:MD=NonUnifSquaredPhys(C/2)}
	\frac{d}{dt}x_{e}(t)=-x_{e}(t)\frac{\partial}{\partial x_{e}}\L(x_{e}(t))\overset{(\ref{eq:defPhi})}{=}\frac{c_{e}}{2}\cdot x_{e}(t)(\twonorm{\Lambda_{e}}^2 -1).
      \end{equation*}
      \proofendswithequation
\end{proof}

As is~\cite{B19}, we can use the connection to mirror descent to estimate the speed of convergence of the Physarum dynamics to the Lyapunov minimum; \cite{B19} builds up on~\cite{Alvarez-Bolte-Brahic,Wilson}.

For a differentiable function $f$ in $m$ variables, the Bregman divergence $D_f$ is a function in $2m$ variables defined by the equation
\[    D_{f}(x,y) =  f(x)-f(y)-\left\langle \nabla f(y),x-y\right\rangle, \]
i.e., as the difference of the function value at $x$ and the value at $x$ of the tangent plane to $f$ at $y$. Clearly, if $f$ is convex,  $D_f$ is non-negative. 

\begin{lemma} Let $h: \R_{\ge 0}^m \rightarrow \R$ be defined by 
\[   h(x)=\sum_e x_{e}\ln x_{e}-\sum_e x_{e}. \]
Then $h$ is convex on $\Rplus^m$, $D_h$ is non-negative, and 
\[ D_h(x,y) = \sum_e x_{e}\ln x_{e}-\sum_e x_{e}\ln y_{e}-\sum_e x_{e}+\sum_e y_{e}, \]
\end{lemma}
\begin{proof} The function $h$ is convex in $x_e$ (partial derivative $\ln x_e$ and second partial derivative $1/x_e$). For its Bregman divergence $D_h$, we compute
\begin{eqnarray*}
D_{h}(x,y) & = & h(x)-h(y)-\left\langle \nabla h(y),x-y\right\rangle \\
		& = & \sum_e x_{e}\ln x_{e}-\sum_e x_{e}- (\sum_e y_{e}\ln y_{e} - \sum_e y_{e}) - \sum_{e}(x_{e} - y_e) \ln y_{e} \\
		& = & \sum_e x_{e}\ln x_{e}-\sum_e x_{e}\ln y_{e}-\sum_e x_{e}+\sum_e y_{e}.
\end{eqnarray*}
So $D_{h}$ is the \emph{relative entropy} function.\end{proof}

\begin{fact}
	\label{fact:PhiConv} \cite[Lemma 2.2]{B19} $\L$ is convex. 
\end{fact}

\begin{theorem} Let $x^*$ be the global minimizer of $\L(x)$. For the dynamics $\dot{x}_e = (c_e/2) \cdot x_e (\twonorm{\row[e]{\Lambda}}^2 - 1)$, we have
	\[
	\L(x(t))\leq\L(x^{*})+\frac{1}{t}D_{h}(x^{*},x(0)).
      \]
      for all $t \ge 0$. 
      In particular, 
      \[ \lim_{t \rightarrow \infty} \L(x(t)) = \L(x^*).\]
\end{theorem}
\begin{proof}
According to~\lref{eq:defPhi} we have
\[ \frac{\partial}{\partial x_{e}}\L(x)=\frac{c_{e}}{2}\left(1- \twonorm{\row[e]{\Lambda}}^2  \right)  \text{ and }
\dot{x}_{e}= x_{e}\left(g_e(\twonorm{\row[e]{\Lambda}}^2) - 1 \right).\]
The time derivative of $D_h(x^*,x(t))$ is given by
\begin{eqnarray*}
\frac{d}{dt}D_{h}(x^{*},x) & = & \frac{d}{dt}\sum_{e=1}^{m}x_{e}^{*}\ln x_{e}^{*}-\frac{d}{dt}\sum_e x_{e}^{*}\ln x_{e}-\frac{d}{dt}\sum_e x_{e}^{*}+\frac{d}{dt}\sum_e x_{e}\\
		& = & \sum_{e=1}^{m}x_{e}^{*}\left(-\frac{1}{x_{e}}\cdot\frac{d}{dt}x_{e}\right)+\sum_e \frac{d}{dt}x_{e}\\
                           &=& \sum_e (x_e - x_e{^*}) \frac{c_e}{2} (\twonorm{\row[e]{\Lambda}}^2 - 1)\\
  &=& - \langle (x - x^*), \nabla \L(x(t))\rangle. 
\end{eqnarray*}
We now consider the function    \renewcommand{\H}{\mathcal{H}}
\[    \H(t)=D_{h}(x^{*},x(t))+t\left[\L(x(t))-\L(x^{*})\right]  \]
Since $\frac{d}{dt}\L(x)\leq0$, by Lemma \ref{lem:GenEqSet}, and $D_{\L}(x^{*},x) \ge 0$ for all $x$, we obtain
	\begin{eqnarray*}
		\frac{d}{dt}\H(t) & = &-\left\langle \nabla\L(x(t)),x(t)-x^{*}\right\rangle + \L(x(t)) -\L(x^{*})+t\cdot\frac{d}{dt}\L(x(t))\\
		& \leq & -\left[\L(x^{*})-\L(x(t))-\left\langle \nabla\L(x),x^{*}-x(t)\right\rangle \right]\\
		& = & -D_{\L}(x^{*},x(t))\\
		& \leq & 0. 
	\end{eqnarray*}
Hence $\H(t)\leq\H(0)$ for all $t\geq0$ and therefore
	\[
	D_h(x^{*},x(t))+t\left[\L(x(t))-\L(x^{*})\right]\leq D_h(x^{*},x(0))+0\left[\L(x(0))-\L(x^{*})\right].
      \]
      and further (using $D_h(x^{*},x(t)) \ge 0$)
	\[
		\L(x(t)) \leq \L(x^{*})+\frac{1}{t}D_h(x^{*},x(0)).
	\]
	
      \end{proof}

      \section{Conclusions}\label{conclusions}

      We proposed a variant of the Physarum dynamics suitable for network design. We exhibited a Lyapunov function for the dynamics, proved convergence of the dynamics, and gave alternative characterizations for the minimum of the Lyapunov function. In the experimental part, we showed that the dynamics captures the positive effect of sharing links and is able to construct \emph{nice} networks.

      Many questions remain open. We do not claim any biological plausibility for our proposal and we have studied one particular form of the dynamics, namely $\dot{x_e} = \abs{q_e} - x_e$. Other dynamics have been studied for the shortest path problem, e.g., $\dot{x_e} = \abs{q_e}^\mu - x_e$ with $\mu > 1$ or $\dot{x_e} = \frac{\abs{q_e}}{a_e + \abs{q_e}} - x_e$~\cite{Miyaji:2008,Meyer-Ansorge-Nakagaki}. The latter paper also studies the influence of noise on the dynamics. An extension to network design would be interesting.

      The papers mentioned in the preceding paragraph are theory papers that investigate variants of the basic dynamics~(\ref{dynamics}). A different line of research aims at a deeper understanding of the inner workings of Physarum polycephalum, for example, how global synchronisation can result from random peristaltics~\cite{Alim2013}, how information can be transported and a memory can exist in an organism without a nervous system~\cite{Alim2017,Alim2021}, and whether tubes of the mold can transfer electricity~\cite{Whiting}.
There seems to be little connection between these lines of research.

      We used an Euler discretization of the dynamics for the experiments in Section~\ref{Case Studies}. The resulting algorithm is quite slow. The Lyapunov function $\L$ is a convex function and hence the tool box of convex optimization is available for computing its minimum. Does this lead to a practical algorithm for network design? \cite{Watanabe-Tero2011} also uses an Euler discretization of the dynamics for their computer experiments. They speed-up the computation by considering only a random subset of the demands instead of all demands in each iteration. If the random subset is not too small, the dynamics seem to converge to the same solution. Is this true generally?

   \newcommand{\htmladdnormallink}[2]{#1}


\newcommand{\etalchar}[1]{$^{#1}$}

\end{document}